\documentclass[a4paper,11pt,reqno]{amsart}
\usepackage{amsmath}
\usepackage[T1]{fontenc}
\usepackage{amssymb}
\usepackage{amsthm}
\usepackage{color}
\usepackage[lmargin=2.5 cm,rmargin=2.5 cm,tmargin=3.5cm,bmargin=2.5cm,
paper=a4paper]{geometry}

\newcommand{\md}{\,{\rm d}}
\newcommand{\Ab}{\mathbf A}
\newcommand{\FF}{\mathbf F}
\newcommand{\R}{\mathbb R}
\newcommand{\C}{\mathbb C}

\DeclareMathOperator{\curl}{curl}\DeclareMathOperator{\Div}{div}
 \DeclareMathOperator{\dist}{dist}

\newtheorem{thm}{Theorem}[section]
\newtheorem{prop}[thm]{Proposition}
\newtheorem{lem}[thm]{Lemma}
\newtheorem{corol}[thm]{Corollary}

\newtheorem{proposition}[thm]{Proposition}

\newtheorem{rem}[thm]{Remark}

\numberwithin{equation}{section}

\title[Bulk superconductivity]{Nucleation of bulk superconductivity
close to critical magnetic field}

\author[S. Fournais]{S\o ren Fournais}
\address[S. Fournais and A. Kachmar]{Department of Mathematical Sciences, University
  of Aarhus, Ny Munkegade, Building
  1530, DK-8000 \AA rhus C, Denmark}
\email[S. Fournais]{fournais@imf.au.dk} \email[A.
Kachmar]{ayman.kachmar@math.u-psud.fr}

\author[A. Kachmar]{Ayman Kachmar}

\date{\today}

\begin{document}

\begin{abstract}
We consider the two-dimensional Ginzburg-Landau functional with
constant applied magnetic field. For applied magnetic fields close
to the second critical field $H_{C_2}$ and large Ginzburg-Landau
parameter, we provide leading order estimates on the energy of
minimizing configurations. We obtain a fine threshold value of the
applied magnetic field for which bulk superconductivity contributes
to the leading order of the energy. Furthermore, the energy of the
bulk is related to that of the Abrikosov problem in a periodic
lattice. A key ingredient of the proof is a novel $L^\infty$-bound
which is of independent interest.
\end{abstract}

\maketitle 

\section{Introduction and main results}\label{hc2-sec:int}
Let us consider a two-dimensional, simply connected, open domain
$\Omega\subset\mathbb R^2$ with smooth boundary. The energy of a
cylindrical superconducting sample of cross section $\Omega$, placed
in a constant applied magnetic field parallel to the cylinder axis,
is given by the following Ginzburg-Landau functional:
\begin{align}\label{eq-hc2-GL}
\mathcal E(\psi,\Ab;\Omega)&=\int_\Omega e_{\kappa,H}(\psi,\Ab)\,dx\nonumber\\
&=\int_\Omega \left(|(\nabla-i\kappa
H\Ab)\psi|^2-\kappa^2|\psi|^2+\frac{\kappa^2}2|\psi|^4+(\kappa
H)^2|\curl(\Ab-\FF)|^2\right)\,\md x\,.
\end{align}
Here $\psi$ is a complex valued wave function, $\Ab:\Omega\to\R^2$ a
vector potential, $\kappa$ the Ginzburg-Landau parameter (a material
parameter which is temperature independent), and  $H$ is the
strength of the applied magnetic field. The potential
$\FF:\Omega\to\R^2$ is the unique vector field satisfying,
\begin{equation}\label{eq-hc2-F}
{\rm curl}\,\FF=1\,,\quad{\Div}\,\FF=0\quad{\rm in}~\Omega\,,\quad
\nu\cdot \FF=0\quad{\rm on}~\partial\Omega\,,
\end{equation}
where $\nu$ is the unit inward normal  vector of $\partial\Omega$.

In the last two decades, many authors have studied the minimization
of the Ginzburg-Landau functional ${\mathcal E}$ in \eqref{eq-hc2-GL}  over all
admissible configurations $(\psi,\Ab)
\in H^1(\Omega;\C)\times H^1(\Omega;\R^2)$.
In the asymptotic limit $\kappa\to\infty$
(corresponding to type II superconductors), it is recognized that
the behavior of the minimizers and their energy strongly depends
on the applied magnetic field $H$. One distinguishes three
different critical values $H_{C_1}$, $H_{C_2}$ and $H_{C_3}$ of the
applied magnetic field that can be described roughly in the
following way:
\begin{enumerate}
\item If the applied magnetic field $H<H_{C_1}$, then $|\psi|$
does not vanish anywhere in $\Omega$, for any minimizer
$(\psi,A)$ of the Ginzburg-Landau energy in \eqref{eq-hc2-GL}.
\item If $H_{C_1}<H<H_{C_2}$, $|\psi|$ has isolated zeros in
$\Omega$, called {\it vortices}.
\item If $H_{C_2}<H<H_{C_3}$, $|\psi|$ is small (in the bulk) except
  in a narrow region near the boundary of $\Omega$. This is the
  phenomenon called boundary superconductivity.
\item If $H>H_{C_3}$, $|\psi|$ vanishes everywhere in $\Omega$.
\end{enumerate}
Precise mathematical definitions exist for the critical fields
$H_{C_1}$ and $H_{C_3}$ which are precisely estimated in the limit
$\kappa\to\infty$. We do not aim at giving an exhaustive list of
references but we invite the reader to see the monographs
\cite{FoHe08, SaSe}. A mathematical definition of the critical field
$H_{C_2}$ is still not available, but current
mathematical results (c.f. \cite{FoHe08, Pa02, SaSe, SS02}) suggest
that it behaves as follows in the large $\kappa$ regime,
$$H_{C_2}=\kappa+o(\kappa)\quad{\rm as~}\kappa\to\infty.$$

The present paper is devoted to a detailed analysis of the  minimizers
of the Ginzburg-Landau functional
in  the asymptotic regime $\kappa\to\infty$ and
$H=\kappa+o(\kappa)$, which corresponds to type~II superconductors
subject to an applied magnetic field $H$ close to the critical field
$H_{C_2}$.
The obtained results  are complementary to those in
\cite{AlHe,  FoHe04, Pa02, SS02}.

\subsection{Earlier results}
The regime of applied magnetic fields  close to the critical field
$H_{C_3}$ is treated by Lu-Pan \cite{LuPa99} (who, in particular,
introduced a precise definition of this critical field), Helffer-Pan \cite{HP} and then by
Fournais-Helffer \cite{FoHe04}. This regime corresponds to applied
magnetic fields $H=\frac{\kappa}{\Theta_0}+\rho(\kappa)$ where
$\rho(\kappa)$
satisfies,
$\lim_{\kappa\to\infty}\frac{\rho(\kappa)}{\kappa}=0$.
The constant $\Theta_0$ appearing above is universal and satisfies
$\Theta_0 \in (0,1)$.

Among other
things, the above mentioned papers give leading order estimates on the
ground state energy,
\begin{equation}\label{eq-GL-GS}
C_0(\kappa,H)=\inf_{(\psi,\Ab)\in H^1(\Omega;\C)\times
  H^1(\Omega;\R^2)}\mathcal E(\psi,\Ab)\,.
\end{equation}

Pan \cite{Pa02} and
Almog-Helffer \cite{AlHe}   give
leading order estimates on the ground state energy, $C_0(\kappa,H)$ when $\kappa\to\infty$ and
the applied magnetic
field satisfies
$$H=b\kappa+o(\kappa)\quad{\rm as~}\kappa\to\infty\,.$$
The constant $b$ is assumed in the interval
$[1,\Theta_0^{-1})$ (with an extra condition when $b=1$, see
Theorem~\ref{thm-hc2-Pan} below). This regime corresponds to applied magnetic fields
varying between the critical fields $H_{C_2}$ and
$H_{C_3}$. Roughly speaking, the above mentioned papers show that the
ground state energy satisfies,
$$C_0(\kappa,H)=-C(b)|\partial\Omega|\kappa+o(\kappa)\quad{\rm
  as~}\kappa\to\infty\,,$$
where $[1,\Theta_0^{-1})\ni b\mapsto C(b)\in (0,\infty)$. The case  $b=1$ corresponds to  applied fields
$H$ close to the critical field $H_{C_2}$.
In strong connection with our results, we state the following theorem proved
by Pan in \cite{Pa02}, devoted to the case $b=1$. We use here the
convention that a set $D\subset\Omega$ is smooth if there exists a
smooth set $\widetilde D\subset\R^2$ such that $D=\widetilde D\cap\Omega$.

\begin{thm}\label{thm-hc2-Pan}
There exists a positive universal constant $E_1$ such that, for any
magnetic field $H=H(\kappa)$ satisfying,
\begin{equation}\label{eq-hc2-hypH}
\frac{H}\kappa\to1\,,\quad
H-\kappa\to+\infty\quad{\rm
as}~\kappa\to\infty\,,\end{equation} any minimizer
$(\psi,\Ab)$ of the energy $\mathcal E$ in \eqref{eq-hc2-GL} and
any open, smooth domain
$D\subset\Omega$, the
following expansion holds
\begin{equation}\label{eq-hc2-asymp}
\mathcal E(\psi,\Ab;D)=- E_1|\overline
D\cap\partial\Omega|\kappa+o(\kappa)\,,\quad{\rm
as}~\kappa\to\infty\,.
\end{equation}
\end{thm}

Furthermore, Pan proves in \cite{Pa02} that  $\psi$ decays
away from the boundary $\partial\Omega$ in the
$L^2$-sense (and this is actually one key ingredient to prove
(\ref{eq-hc2-asymp})) showing thus that the superconducting sample
exhibits only {\it surface} superconductivity. A result by Almog
\cite{Al2} on the decay of $\psi$  permits one to extend the
validity of Theorem~\ref{thm-hc2-Pan} down to magnetic fields $H$
satisfying $H-\kappa\gg\frac{\ln\kappa}{\kappa}$ as
$\kappa\to\infty$. Here we remind the reader that for two positive
functions $a(\kappa)$ and $b(\kappa)$, the notation $a(\kappa)\ll
b(\kappa)$ as $\kappa\to\infty$ means that
$\displaystyle\lim_{\kappa\to\infty} \frac{a(\kappa)}{b(\kappa)}=0$.

The constant $E_1$ appearing in Theorem~\ref{thm-hc2-Pan} is a
universal constant defined via a reduced Ginzburg-Landau energy in a
cylindrical domain. We will recall its definition in \eqref{eq-hc2-E1}
below.

Complementary to the results of Pan \cite{Pa02},
Sandier and Serfaty \cite{SS02}
consider the regime of magnetic fields
$$H=b\kappa+o(\kappa)\quad{\rm as~}\kappa\to\infty\,,$$
where the constant
$b\in[0,1]$. Among other things, they prove that there exists a
strictly increasing function $[0,1]\ni b\mapsto f(b)\in [-\frac12,0]$,
with $f(1)=0$,
such that the ground state energy satisfies
$$C_0(\kappa,H)=f(b)|\Omega|\kappa^2+o(\kappa^2)\quad{\rm
  as~}\kappa\to\infty\,.$$
More precisely, they prove a uniform energy density in $|\Omega|$
compatible with this global ground state energy. In the regime of
interest to us, which corresponds to $b=1$, the ground state energy
therefore 
satisfies,
\begin{equation}\label{en-SS}
C_0(\kappa,H)=o(\kappa^2)\quad{\rm as~}\kappa\to\infty\,.
\end{equation}

We observe from the aforementioned results that a transition happens
from {\it bulk} to {\it boundary} behavior when the applied field is
close to $\kappa$, or in other words, when the applied field is close
to the second critical field $H_{C_2}$.
At the same time the order of magnitude of the energy changes here.

The results
of the present paper (Theorem~\ref{thm-hc2-FK1} below)
determine the leading order term in the
energy expansion \eqref{en-SS}, and indicates the optimal regime for
the magnetic field $H$ such that Theorem~\ref{thm-hc2-Pan} is valid.
We obtain that the leading order behavior of the energy is
determined according to variations of $H-\kappa$ on the
critical scale $\sqrt{\kappa}$. Our results close the gap between the results of \cite{Pa02}
and \cite{SS02} and---taken together with the results of these papers---yield an overall
understanding of the ground state energy of type II superconductors in
strong magnetic fields.

\subsection{Main results}

In addition to the constant $E_1$ appearing in
Theorem~\ref{thm-hc2-Pan}, the asymptotic behavior of the ground state
energy $C_0(\kappa,H)$
involves another universal constant $E_2>0$. The definition of
$E_2$ is
related
to the Abrikosov energy, see \eqref{eq-hc2-E2-TDL} and
\eqref{eq-hc2-E2} below. We will use the function $\R\ni
x\mapsto[x]_+:= {\rm max}(0,x)$. 

\begin{thm}\label{thm-hc2-FK1}
Let the positive constants $E_1$ and $E_2$ be defined by
\eqref{eq-hc2-E1} and \eqref{eq-hc2-E2-TDL} respectively. Assume that the
magnetic field satisfies,
$$H=\kappa-\mu(\kappa)\sqrt{\kappa}\quad{\rm such ~that~}
\quad \lim_{\kappa\to\infty}\frac{\mu(\kappa)}{\sqrt{\kappa}}=0
\,.$$ 
Then, for any minimizer $(\psi,\Ab)$ of the energy $\mathcal E$ in
\eqref{eq-hc2-GL}, and
any open, smooth domain $D\subset\Omega$, the following asymptotic expansion
holds,
\begin{equation}\label{eq-hc2-asymp'}
\mathcal E(\psi,\Ab;D)=-E_1|\overline
D\cap\partial\Omega|\kappa-E_2|D|\,[\mu(\kappa)]^2_+\kappa+
o(\max(1,[\mu(\kappa)]^2_+)\kappa)\,,\quad{\rm as}~\kappa\rightarrow\infty\,.
\end{equation}
\end{thm}

Theorem~\ref{thm-hc2-FK1} generalizes  Theorem~\ref{thm-hc2-Pan}
and shows,
in an energy sense, that the sample
is in a surface superconducting state as long as the magnetic field
satisfies $|H-\kappa|\ll\sqrt{\kappa}$ (see
Corollary~\ref{corol-hc2-FK1} for a qualitative statement on the
behavior of order parameters $\psi$). In this specific regime,
one difference between the proofs of
Theorems~\ref{thm-hc2-Pan} and \ref{thm-hc2-FK1} is that the
order parameter $\psi$ is not expected to decay in the bulk. Hence we
need a different method for controlling the energy contribution of
the bulk, which we show to be negligible compared with that of the
boundary.

However, as Theorem~\ref{thm-hc2-FK1} shows, when the magnetic field
strength $H$ becomes of the order  $\kappa-\mu\sqrt{\kappa}$ with
$\mu$ a positive constant, the energy contribution of the bulk can no
more be neglected. Therefore, Theorem~\ref{thm-hc2-Pan} gives a sharp
description of how bulk superconductivity starts to appear,
and thus establishes  a fine characterization of the critical
field $H_{C_2}$, which seems to be absent even in the Physics
literature.

When the difference $\kappa-H$ becomes large compared with the
critical scale $\sqrt{\kappa}$\,, Theorem~\ref{thm-hc2-FK1} shows
that the energy of the bulk becomes
dominant to leading order.

As a corollary of Theorem~\ref{thm-hc2-FK1},
we get the following properties of the minimizing order parameter.

\begin{corol}\label{corol-hc2-FK1}
Assume that the magnetic field satisfies,
$$H=\kappa-\mu(\kappa)\sqrt{\kappa}\,,\quad{\rm such~ that~}
\lim_{\kappa\to\infty} \frac{\mu(\kappa)}{\sqrt{\kappa}}=0\,.
$$
Then for any minimizer $(\psi,\Ab)$ of the energy $\mathcal E$ in
\eqref{eq-hc2-GL}, and any
open, smooth domain  $D\subset\Omega$,
we have
\begin{align}
  \label{eq:14}
  \kappa\int_{D}|\psi|^4\,dx = 2\left(E_1|\overline D\cap\partial\Omega|
 +[\mu(\kappa)]^2_+E_2|D|\right) + o(\max(1,[\mu(\kappa)]_{+}^2))\,.
\end{align}
\end{corol}

We conclude by stating a sharp $L^\infty$-bound in the following
theorem. The motivation for this is twofold.
Taken together with \cite[Theorem~2.1]{FH08}, it is
an affermative
answer to a precise version
of
a conjecture by Sandier-Serfaty
\cite{SS02} and Aftalion-Serfaty \cite{AS}. It also plays a key-role
in the proof of Theorem~\ref{thm-hc2-FK1} announced above.

\begin{thm}
\label{thm:Linfty-bulk} Let $\delta \in (0,1)$ and $g:\R_+\to\R_+$
be a function such that $g(\kappa)/\kappa \rightarrow 0$ as $\kappa\to\infty$.
Then there exists a constant $C>0$ such that if $|H - \kappa|\leq g(\kappa)$, then
\begin{align}
\| \psi \|_{L^{\infty}(\omega_{\kappa})} \leq C \lambda(\kappa),
\end{align}
for all critical points $(\psi, {\bf A})$ of the energy in
\eqref{eq-hc2-GL}.

Here
\begin{align}
\omega_{\kappa} := \{ x \in \Omega\,|\, \dist(x,\partial \Omega)
\geq \kappa^{-1+\delta}\},
\end{align}
and
\begin{align}
\lambda(\kappa) := \max\left\{ \left|\frac{\kappa}{H} -
    1\right|^{1/2},
\kappa^{-1+\delta}\right\}.
\end{align}
\end{thm}

In the regime of applied fields $H=\kappa-\mu(\kappa)\sqrt{\kappa}$
with
$\displaystyle\lim_{\kappa\to\infty}\frac{\mu(\kappa)}{\sqrt{\kappa}}=0$
and $\displaystyle\lim_{\kappa\to\infty}\mu(\kappa)=\mu_0\in
(0,+\infty]$, the estimate of Theorem~\ref{thm:Linfty-bulk} is
optimal. In this regime, the constant $\lambda(\kappa)$
above is equal to $\left|\frac{\kappa}{H}-1\right|^{1/2}$. It
follows from Corollary~\ref{corol-hc2-FK1} that there exists a
constant $c>0$ such that, for any minimizer $(\psi,\Ab)$, we have,
$$c\left|\frac{\kappa}{H}-1\right|^{1/2}\leq \|\psi\|_{L^\infty(\omega_\kappa)}\,.$$

The paper is organized as follows.
Section~\ref{hc2-sec-preliminaries} is devoted to some
preliminaries, in particular, we recall some {\it a
  priori} estimates together with the definitions of
the universal constants $E_1$ and $E_2$.
 The proof of
Theorem~\ref{thm:Linfty-bulk} is given in
Section~\ref{sec-hc2-Linfty}.
In Sections~\ref{hc2-sec-ub} and
\ref{hc2-sec-lb},
matching upper and lower bounds for the functional
in \eqref{eq-hc2-GL} are obtained.
Finally, Section~\ref{sec-hc2-proofs} concludes with
the proof of Theorem~\ref{thm-hc2-FK1}.

\section{Preliminaries}\label{hc2-sec-preliminaries}
\subsection{A priori estimates}\label{hc2-sec:ape} In this
section, we collect some useful estimates for critical points of the
Ginzburg-Landau functional $\mathcal E$ introduced in
(\ref{eq-hc2-GL}).
The set of estimates
in Lemma~\ref{lem-hc2-FoHe} appeared first in \cite{LuPa99} (for a
more particular regime) and were then  proved for a wider regime in
\cite{FoHe08, FoHe06}. The estimate of Lemma~\ref{lem-FH-jems} was
proved recently in \cite{FH08}.

Notice that a critical point $(\psi,\Ab)$ of the functional
$\mathcal E$ is a solution of the Ginzburg-Landau equations:
\begin{equation}\label{eq-hc2-GLeq}
\left\{
\begin{array}{l}
-(\nabla-i\kappa H\Ab)^2\psi=\kappa^2(1-|\psi|^2)\psi\,,\\
-\nabla^\bot\curl\Ab=(\kappa H)^{-1}{\rm
Im}(\overline\psi\,(\nabla-i\kappa
H\Ab)\psi)\,,\quad{\rm in}~\Omega\,,\\
\nu\cdot(\nabla-i\kappa H\Ab)\psi=0\,,\quad \curl\Ab=1\,,\quad{\rm
  on}~\partial\Omega\,.
\end{array}\right.
\end{equation}
Here $\nu$ is the unit inward normal vector of $\partial\Omega$.

We start with the following direct consequence of the maximum
principle.

\begin{lem}\label{lem-hc2-maxp}(\cite[Chapter~3]{SaSe})
Let $(\psi,\Ab)$ be a solution  of (\ref{eq-hc2-GLeq}). Then
$|\psi|\leq 1$ in $\overline\Omega$.
\end{lem}

We also have elliptic estimates on the magnetic field and the energy density.

\begin{lem}\label{lem-hc2-FoHe}(Fournais-Helffer \cite{FoHe06})
There exist positive constants $\kappa_0$ and $C$ such that, if the
magnetic field satisfies $H\geq\frac\kappa2$ and if $(\psi,\Ab)$ is
a critical point of (\ref{eq-hc2-GL}), then for all
$\kappa\geq\kappa_0$, the following estimates hold,
\begin{align}\label{eq-hc2-FoHe1}
&\|\curl(\Ab-\FF)\|_{C^1(\Omega)}+\kappa^{-1}\|\curl(\Ab-\FF)\|_{C^2(\Omega)}\leq
C\kappa^{-1}\,,\\
&\|(\nabla-i\kappa H\Ab)\psi\|_{L^\infty(\Omega)}\leq
C\kappa\,,\quad e_{\kappa,H}(\psi,\Ab)\leq
C\kappa^2\,.\label{eq-hc2-FoHe2}
\end{align}
\end{lem}

Finally, close to $H_{C_2}$ the estimate of Lemma~\ref{lem-hc2-maxp}
can be improved.

\begin{lem}\label{lem-FH-jems}(Fournais-Helffer \cite{FH08})
Assume that the magnetic field $H=H(\kappa)$ satisfies
$\frac{H}{\kappa}\to1$
as $\kappa\to\infty$. Then, given any function $g_1:\mathbb
R_+\to(0,1]$ satisfying
$$\lim_{\kappa\to\infty}g_1(\kappa)=0\,,\quad\lim_{\kappa\to\infty}\kappa
g_1(\kappa)=\infty\,,$$ there exists a function $g_2:\mathbb
R_+\to(0,1]$ such that
$$\lim_{\kappa\to\infty}g_2(\kappa)=0$$
and
\begin{equation}\label{eq-FH-jems}
\|\psi\|_{L^\infty(\{x\in\Omega~:~{\rm dist}(x,\partial\Omega)\geq
g_1(\kappa)\})}\leq g_2(\kappa)\,.
\end{equation}
\end{lem}

\subsection{The limiting boundary problem}\label{hc2-sec-lbdryp}
We recall in this section the definition of the universal constant
$E_1$ (appearing in Theorem~\ref{thm-hc2-FK1} as given in \cite{Pa02}).

Let us consider the following magnetic potential (we keep the
notation of \cite{Pa02}),
\begin{equation}\label{eq-hc2-mp:E}
\mathbf E(x)=(-x_2,0)\,,\quad\forall~x=(x_1,x_2)\in\mathbb
R\times\mathbb R_+\,,
\end{equation}
together with the reduced Ginzburg-Landau energy,
\begin{equation}\label{eq-hc2-redGL}
\mathcal E_\ell(\phi)=\int_{U_\ell}\left(|(\nabla-i\mathbf
E)\phi|^2-|\phi|^2+\frac12|\phi|^4\right)\,dx\,,
\end{equation}
where $U_\ell$ is the domain,
\begin{equation}\label{eq-hc2-Uell}
U_\ell=(-\ell,\ell)\times(0,\infty)\,,\quad\ell>0\,.
\end{equation}
Let us introduce the space
\begin{equation}\label{eq-hc2-Confspace}
\mathcal V(U_\ell)=\{u\in L^2(U_\ell)~:~(\nabla-i\mathbf E)u\in
L^2(U_\ell)~,~u(\pm\ell,\cdot)=0\,\}\,.
\end{equation}
We are interested in minimizing the energy (\ref{eq-hc2-redGL}) over
the space $\mathcal V(U_\ell)$. So we introduce further,
\begin{equation}\label{eq-hc2-d(ell)}
d(\ell)=\inf\{\mathcal E_\ell(\phi)~:~\phi\in\mathcal V(U_\ell)\}\,.
\end{equation}
The following theorem is proved in \cite[Theorems~4.4 \& 5.3]{Pa02}.
\begin{thm}\label{thm-hc2-Pa02}
There exist positive constants $\ell_0$, $M$ and $E_1$ such that:
\begin{enumerate}
\item For all $\ell\geq \ell_0$, (\ref{eq-hc2-redGL}) has a
minimizer $\phi_\ell$ in $\mathcal V(U_\ell)$, and
$\phi_\ell\not\equiv0$\,.
\item For all $\ell\geq\ell_0$, $\phi_\ell$ decays in the following way,
$$\int_{U_\ell\cap\{x_2\geq 3\}}\frac{x_2^2}{\ln
x_2}\left(|(\nabla-i\mathbf
E)\phi_\ell|^2+|\phi_\ell|^2+x_2^2|\phi_\ell|^4\right)\,dx\leq
M\ell\,.$$
\item For all $\ell\geq\ell_0$, the following estimate holds
$$\left|\frac{d(\ell)}{2\ell}+E_1\right|\leq\frac{M}{\ell}\,.$$
\end{enumerate}
\end{thm}

In light of Theorem~\ref{thm-hc2-Pa02}, the universal constant
$E_1>0$ is actually given as the limit,
\begin{equation}\label{eq-hc2-E1}
E_1=\lim_{\ell\to\infty}\left(-\frac{d(\ell)}{2\ell}\right)\,.\end{equation}

\subsection{The limiting bulk problem}\label{hc2-sec-lbulkp}
We turn now to the limiting problem in the bulk, thereby defining
the constant $E_2$ appearing in (\ref{eq-hc2-asymp'}). Actually,
$E_2$ can be defined in two different ways. The simpler definition
is through a  ``thermodynamic limit'' of the Ginzburg-Landau energy
(see \eqref{eq-hc2-E2-TDL}). A more complicated  definition is via a
limiting Abrikosov energy in a periodic lattice (see
\eqref{eq-hc2-E2}). The latter approach in defining $E_2$ has more
advantages, since on the one hand it shows rigorously how the
Abrikosov energy links to the Ginzburg-Landau model, and on the
other hand it provides an essential key for proving the main
theorem of the present paper.

\subsubsection{The universal constant $E_2$.}
Let us consider a constant $b\in (0,1)$. For any domain $\mathcal
D\subset \R^2$, we define the following Ginzburg-Landau energy,
$$G_{\mathcal D}(u)=\int_{\mathcal D}b|(\nabla-i\Ab_0)u|^2
-|u|^2+\frac1{2}|u|^4\,dx\,.
$$
Here $\Ab_0$ is the canonical magnetic potential (we keep the
notation from \cite{AS}),
\begin{equation}\label{eq-hc2-mpA0}
\Ab_0(x_1,x_2)=\frac12(-x_2,x_1)\,,\quad\forall~x=(x_1,x_2)\in
\R^2\,.\end{equation} It is proved by Sandier and Serfaty
\cite{SS02} (see also Aftalion-Serfaty \cite[Lemma~2.4]{AS}) that
there exists a continuous increasing function
$g:(0,1]\to(-\frac12,0]$ such that the following identity  holds,
$$g(b)=\lim_{R\to\infty}\frac{\inf_{u\in H^1_0(K_R;\C)}
G_{K_R}(u,A)}{|K_R|}\,,
$$
where $K_R\subset\R^2$ is a square of side-length equal to $R$.
Furthermore, it is proved that there exists a constant
$\alpha\in(0,\frac12)$ such that
$$\alpha(b-1)^2\leq |g(b)|\leq \frac12(b-1)^2\,,\quad\forall~b\in (0,1)\,.$$

The universal constant $E_2$ is then defined by,
\begin{equation}\label{eq-hc2-E2-TDL}
E_2=\lim_{b\to1_-}\frac{|g(b)|}{(b-1)^2}\,.
\end{equation}
The existence of the limit above is proved in \cite[Theorem~2]{AS} and
clearly satisfies
$$0<\alpha\leq E_2\leq \frac12\,.$$

\begin{rem}
For the sake of simplicity
we considered only a square lattice above.
This is because the lattice geometry is not important for the energy
at this level.
In \cite{AS}, the results
above are shown to be true for any parallelogram lattice and with the
same constant $E_2$. This remark also applies to the remainder of the
paper: We work with a square lattice as the basis for our
constructions out of simplicity, and since this is known not to affect
the energy to the precision considered.
\end{rem}

\begin{rem}
Notice that the functional $G_{\mathcal D}$ can be rewritten, using
the simple change of function $u = \sqrt{1-b}\, v$, as follows,
\begin{align*}
  G_{\mathcal D}(u)=(1-b)^2\Big\{ \frac{b}{1-b}\int_{\mathcal D}|(\nabla-i\Ab_0)v|^2
-|v|^2\,dx +\int_{\mathcal D} \frac1{2}|v|^4 - |v|^2\,dx\Big\}\,.
\end{align*}
This simple manipulation provides a link between the Ginzburg-Landau
energy $G_{\mathcal D}$ and the Abri\-kosov energy of
Theorem~\ref{thm-AS} below.
\end{rem}

\subsubsection{The periodic Schr\"odinger operator with constant
  magnetic field.}

Let $R>0$ and denote by $K_{R}$ the unit parallelogram of the
lattice $\mathcal L_{R}=R(\mathbb Z\oplus i\mathbb Z)$. We assume
the quantization condition that $|K_{R}|/(2\pi)$ is an integer, i.e.
there exists $N\in\mathbb N$ such that,
\begin{equation}\label{eq-hc2-quantization}
R^2=2\pi N\,.\end{equation} Let us introduce the following space,
\begin{align}\label{eq-hc2-space1}
E_{R}=\bigg{\{}u\in H^1(K_{R};\C)~:~&u(z_1+R,z_2)=e^{i\frac{\pi
Nz_2}{R }}u(z_1,z_2)\nonumber\\
&u(z_1,z_2+R )=e^{-i\frac{\pi Nz_1}{R } }u(z_1,z_2)\bigg{\}}\,.
\end{align}
Recall the  magnetic potential $\Ab_0$ introduced in
\eqref{eq-hc2-mpA0} above.
 Notice that the periodicity conditions in
\eqref{eq-hc2-space1} are constructed in such a manner that, for any
function $u\in E_{R}$, the functions $|u|$, $|\nabla_{\Ab_0}u|$ and
the vector field
$\overline u \nabla_{\Ab_0}u$ are periodic with respect to
the lattice $K_{R}$.

We denote by $P_{R}$ the operator,
\begin{equation}\label{eq-hc2-poperator}
P_{R}=-(\nabla-i\Ab_0)^2\quad{\rm in}~L^2(K_{R})\,,
\end{equation}
with form domain the space $E_{R}$ introduced in
(\ref{eq-hc2-space1}). More precisely, $P_{R}$ is the self-adjoint
realization associated with the closed quadratic form
\begin{equation}\label{eq-hc2-poperatorQF}
E_{R}\ni f\mapsto Q_{R}(f)=\|(\nabla-i\Ab_0)f\|_{L^2(K_{R})}^2\,.
\end{equation}

The operator $P_{R}$ being with compact resolvent, let us denote by
$\{\mu_j(P_{R})\}_{j\geq1}$ the increasing sequence of its distinct
eigenvalues (i.e. without counting multiplicity).

The following proposition may be classical in the spectral theory
of Schr\"odinger operators, but we refer to \cite{AS} or \cite{Al}
for a simple proof.

\begin{prop}\label{prop-hc2-poperator}
Assuming $R$  is such that $|K_{R}|\in2\pi\mathbb N$, then the
operator $P_{R}$ enjoys the following spectral properties:
\begin{enumerate}
\item $\mu_1(P_{R})=1$ and $\mu_2(P_{R})\geq 3$\,.
\item The space $L_{R}={\rm Ker}(P_{R}-1)$ is finite
dimensional and ${\rm dim}\,L_{R}=|K_{R}|/(2\pi)$\,.
\end{enumerate}
Consequently, denoting by $\Pi_1$ the orthogonal projection on the
space $L_{R}$ (in $L^2(K_{R})$), and by $\Pi_2={\rm Id}-\Pi_1$, then
for all $f\in D(P_{R})$, we have,
$$\langle P_{R}\Pi_2 f\,,\,\Pi_2f\rangle_{L^2(K_{R})}\geq
3\|f\|^2_{L^2(K_{R})}\,.$$
\end{prop}

The next lemma is a consequence of the existence of a spectral gap
between the first two eigenvalues of $P_{R}$.

\begin{lem}\label{prop-hc2-poperator'}
Given $p\geq 2$, there exists a constant $C_p>0$ such that, for any
$\gamma\in(0,\frac12)$, $R\geq 1$ with $|K_R| \in 2 \pi {\mathbb N}$, and $f\in D(P_{R})$ satisfying
\begin{equation}\label{eq-hc2-hypf}
Q_{R}(f)-(1+\gamma)\|f\|^2_{L^2(K_{R})}\leq0\,,\end{equation} the
following estimate holds,
\begin{equation}\label{eq-hc2-1=proj}
\|f-\Pi_1f\|_{L^p(K_{R})}\leq
C_p\sqrt{\gamma}\,\|f\|_{L^2(K_{R})}\,.
\end{equation}
Here $\Pi_1$ is the projection on the space $L_{R}$.
\end{lem}
\begin{proof}
Let us write $f_1=\Pi_1f$ and $f_2=f-\Pi_1f$, then since $f_1$ and
$f_2$ are orthogonal we get ($\|\cdot\|$ denotes the $L^2$ norm
unless otherwise stated),
$$Q_{R}(f)=Q_{R}(f_1)+Q_{R}(f_2)\,,\quad \|f\|^2=\|f_1\|^2+\|f_2\|^2\,.$$
Furthermore, from (\ref{eq-hc2-hypf}) we deduce,
$$\gamma\|f\|^2\geq
Q_{R}(f_1)-\|f_1\|^2+Q_{R}(f_2)-\|f_2\|^2 =Q_{R}(f_2)-\|f_2\|^2 \,.
$$ Invoking Proposition~\ref{prop-hc2-poperator} and the min-max
variational principle, we infer the bound,
\begin{align}
\label{eq-hc2-1=proj'}
\gamma\|f \|^2\geq
\frac12Q_{R}(f_2)+\frac12\|f_2\|^2\,.
\end{align}
Now, we claim that the following Sobolev inequality holds,
\begin{equation}\label{eq-hc2-H1-Lp}\|f_2\|_{L^p(K_{R})}\leq
C_p\left(\|\nabla|f_2|\,\|_{L^2(K_{R})}+\|f_2\|_{L^2(K_{R})}\right)\,,
\end{equation}
where $C_p>0$ is a constant independent of $R\in[1,\infty)$.

Using the diamagnetic inequality, we get further,
$$\|f_2\|_{L^p(K_{R})}\leq
C_{p}\left(\sqrt{Q_{R}(f_2)}+\|f_2\|_{L^2(K_{R})}\right)\,.
$$
By implementing \eqref{eq-hc2-1=proj'} in the above estimate, we get
the estimate announced in the lemma.

Thus, to finish the proof, we need only establish the estimate
\eqref{eq-hc2-H1-Lp}. Let $\chi$ be a cut-off function such that
$0\leq\chi\leq1$ in $\R^2$, $\chi=1$ in $B(0,1)$ and ${\rm
supp}\chi\subset B(0,2)$. Let further $C$ be a positive constant
such that $B(0,C)$ contains $K_{1}$.

The function
$$g(x)=\chi\left(\frac{x}{CR}\right)\,|f_2(x)|\,,\quad x\in\R^2\,,$$
belongs now to $H^1(\R^2)$. Using the Sobolev embedding of
$H^1(\R^2)$ in $L^p(\R^2)$, $p\geq 2$, we get a constant $c_p>0$
such that
$$\|g\|_{L^p(\R^2)}\leq c_p\left(\|\nabla
  g\|_{L^2(\R^2)}+\|g\|_{L^2(\R^2)}\right)\,.$$
Since the function $|f_2|$ is periodic with respect to the lattice
$K_{R}$, and since
$$\|\nabla g\|_{L^2(\R^2)}^2\leq 2\|\nabla |f_2|\,\|_{L^2(B(0,CR)}
+\frac{2}{C^2R^2}
\|f_2\|_{L^2(B(0,CR)}^2\,,$$
we get the estimate in
\eqref{eq-hc2-H1-Lp}.
\end{proof}

\subsubsection{The Abrikosov energy.}
Let us now introduce the following  energy functional (the Abrikosov
energy),
\begin{equation}\label{eq-hc2-eneAb}
F_{R}(v)=\frac1{|K_{R}|}\int_{K_{R}}\left(\frac12|v|^4-|v|^2\right)\,d
x\,.
\end{equation}
The energy $F_R$ will be minimized on the space
$L_{R}$, the eigenspace of the first eigenvalue of the periodic
operator $P_{R}$,
\begin{align}\label{eq-hc2-space2}
L_{R}&=\{u\in E_{R}~:~P_{R}u=u\}\nonumber\\
&=\{u\in
E_{R}~:~\left(\partial_{x_1}+i\partial_{x_2}+\frac12(x_1+ix_2)\right)u=0\}\,.
\end{align}

The following theorem is proved in \cite[Theorems~1 \& 2]{AS}.

\begin{thm}\label{thm-AS}
Setting
\begin{equation}\label{eq-hc2-c(r,t)}
c(R)=\min\{F_{R}(u)~:~u\in L_{R}\}\,,
\end{equation}
then the limit
$$\lim_{\substack{R\to \infty\\|K_{R}|/(2\pi)\in\mathbb
    N}}c(R)$$ exists and
is expressed  by the universal constant $E_2$ introduced in
\eqref{eq-hc2-E2-TDL} as follows,
\begin{equation}\label{eq-hc2-E2}
E_2=\lim_{\substack{R\to\infty\\|K_{R}|/(2\pi)\in\mathbb
N}}[-c(R)]\,.\end{equation}
\end{thm}

We conclude by showing that \eqref{eq-hc2-eneAb} admits minimizers
in \eqref{eq-hc2-space2}.

\begin{prop}\label{prop-hc2-exmin}
Let $F_{R}$ be the energy introduced in \eqref{eq-hc2-eneAb}. The
infimum of $F_{R}$  over the (eigen-) space $L_{R}$ is achieved by a
function $f_{R}\in L_{R}$.

Furthermore, there exist positive constants  $R_0$ and $C$ such
that, for all $\tau\in\C\setminus\R$ and $R\geq R_0$, we have the
estimate,
\begin{equation}\label{eq-hc2-dec:ftr}
\frac1{|K_{R}|}\int_{K_{R}}|f_{R}|^2\,dx+
\frac1{|K_{R}|}\int_{K_{R}}|f_{R}|^4\, dx\leq C\,.
\end{equation}
\end{prop}
\begin{proof}
The functional $F_R$ is clearly continuous on the finite dimensional
space $L_R$. By applying the Cauchy-Schwarz inequality twice, we notice that,
\begin{align*}
F_R(v)&\geq
\frac1{|K_R|}\left(\frac12\int_{K_R}|v|^4\,dx-\sqrt{|K_R|}\left(\int_{K_R}|v|^4\,dx\right)^{1/2}\right)\\
&\geq \frac1{|K_R|}\left(\frac14\int_{K_R}|v|^4\,dx-10|K_R|\right)\,,
\quad\forall~v\in L_R\,.
\end{align*}  Hence, $F_R$ is positive outside a compact set
and therefore the the (negative) minimum exists in the finite dimensional space
$L_R$.

Noticing that $F_R(f_R)\leq0$, we get the estimate
\eqref{eq-hc2-dec:ftr} from the aforementioned Cauchy-Schwarz inequality.
\end{proof}

\section{The improved $L^\infty$-bound}\label{sec-hc2-Linfty}
This section is devoted to the proof of
Theorem~\ref{thm:Linfty-bulk}. Before we give the proof, we state
the following corollary to Theorem~\ref{thm:Linfty-bulk}, which will
be a key-ingredient in proving Theorem~\ref{thm-hc2-FK1}.

\begin{corol}\label{cor-hc2:L2}
Assume that the magnetic field satisfies
$H=\kappa+\nu$ with $|\nu|\ll\kappa$ as $\kappa\to\infty$.
Then there exist constants $C>0$ and $\kappa_0>0$ such that
$$\|\psi\|_{L^2(\Omega)}\leq
C\zeta(\kappa)\,,\quad\forall~\kappa\geq\kappa_0\,,$$ for all
critical points $(\psi,\Ab)$ of the energy $\mathcal E$ in
\eqref{eq-hc2-GL}.
Here
$$\zeta(\kappa)=\max\left\{\left|1-\frac\kappa{H}\right|^{1/2},
\kappa^{-1/4}\right\}\,.$$
\end{corol}
\begin{proof}
Let $\delta=\frac12$. We write,
\begin{align*}
\int_\Omega|\psi|^2\,dx&=\int_{\{{\rm dist}(x,\partial\Omega)\leq \kappa^{-1+\delta}\}}|\psi|^2\,dx+\int_{\{{\rm dist}(x,\partial\Omega)\geq \kappa^{-1+\delta}\}}|\psi|^2\,dx\\
&\leq c\kappa^{-1+\delta}\|\psi\|_{L^\infty(\Omega)}+C\|\psi\|_{L^\infty(w_\kappa)}\,.
\end{align*}
Here $C$ is a positive constant
and $w_\kappa=\{x\in\Omega~:~{\rm dist}(x,\partial\Omega)\geq
\kappa^{-1+\delta}\}$. Invoking the bound $|\psi|\leq 1$ together with
the estimate in Theorem~\ref{thm:Linfty-bulk} and our choice of
$\delta=\frac12$, we
get for some new constant $C$,
$$\int_\Omega|\psi|^2\,dx\leq C\left(\kappa^{-1/2}+\max\left\{\left|1-\frac{\kappa}{H}\right|,\kappa^{-1}\right\}\right)
\,,$$
which is the bound we wanted to prove.
\end{proof}

Now we proceed in proving Theorem~\ref{thm:Linfty-bulk}.

\begin{proposition}\label{prop:CurlLinfty}
Let $\delta\in(0,1)$. There exist positive constants $\kappa_0$ and
$C$ such that if $H \geq \kappa/2$, $\kappa \geq \kappa_0$, then
\begin{align}
  \label{eq:2}
  \| \curl {\bf A} - 1 \|_{L^{\infty}(\Omega)} \leq C H^{-1}
  (\kappa^{-1+\delta} + \| \psi \|_{L^{\infty}(\omega_{\kappa})})\,.
\end{align}
\end{proposition}

\begin{proof}
Since $\curl {\bf A} = 1$ on $\partial \Omega$, we get by
integrating from the boundary and using the second Ginzburg-Landau
equation in \eqref{eq-hc2-GLeq},
\begin{align}
  \label{eq:3}
  |\curl {\bf A}(x) - 1| &\leq (\kappa H)^{-1} \big\{\kappa^{-1+\delta} +
  \dist(x,\partial \Omega) \| \psi
  \|_{L^{\infty}(\omega_{\kappa})}\big\} \| \nabla - i\kappa H {\bf A}
  \psi \|_{L^{\infty}(\Omega)}\nonumber \\
&\leq \frac{C}{H} \big\{\kappa^{-1+\delta} +
   \| \psi \|_{L^{\infty}(\omega_{\kappa})}\big\}\,,
\end{align}
where we used \eqref{eq-hc2-FoHe2} to get the second inequality.
\end{proof}

\begin{proof}[Proof of Theorem~\ref{thm:Linfty-bulk}]~\\
 We argue by contradiction. Assume that there exist sequences
 $\{\kappa_n\}$, $\{H_n\}$ and a sequence of critical points $\{(\psi_n,\Ab_n)\}$ such that,
 $$\kappa_n\to\infty\,,\quad H_n=\kappa_n + \nu_n,\quad \text{ where } \quad |\nu_n| \leq g(\kappa_n)\,,$$ and
\begin{equation}\label{hc2-FK-infty-hyp1}
\lambda_n^{-1} \|\psi_n\|_{L^\infty(\omega_n)}\to\infty\quad{\rm as
}~n\to\infty\,.\end{equation}
Here we have simplified notation by defining
\begin{align}
  \label{eq:1}
  \omega_n := \omega_{\kappa_n},\qquad \lambda_n := \lambda(\kappa_n).
\end{align}
Since $\lambda_n \geq \kappa_n^{-1+\delta}$, we have
\begin{align}
\label{eq:powerexp} 0<\lambda_n^{-1}
2^{-\kappa_n^{\delta/2}}\leq(\kappa_n^{\delta/2})^{2(\frac1\delta-1)}
2^{-\kappa_n^{\delta/2}}\to0\quad{\rm as }~n\to\infty.
\end{align}
Let us pick $N_0\in\mathbb N$ sufficiently large such that, for all
$n\geq N_0$ we have,
\begin{equation}\label{hc2-FK-infty-hyp2}
\lambda_n^{-1}\|\psi_n\|_{L^\infty(\omega_n)}\geq2 \quad{\rm
and}\quad \lambda_n^{-1}2^{-\kappa_n^{\delta/2}}\leq 1\,.
\end{equation}

For $M>0$ and $n\in\mathbb N$, we define
$$\omega_{M,n}=\{x\in \Omega~:~{\rm dist}(x,\omega_n)\leq
\frac{M \kappa_n^{\delta/2}}{2 \kappa_n}\}\,.$$ We claim that for
each $n\geq N_0$, there exists $M_n\in[0,\kappa_n^{\delta/2}] \cap
{\mathbb N}$ such that,
\begin{align}
    \label{eq:52}
    \| \psi_n \|_{L^{\infty}(\omega_{M_n,n})} \geq \frac{1}{2} \| \psi_n \|_{L^{\infty}(\omega_{M_n+1,n})}.
  \end{align}
Otherwise,  there exists $n\geq N_0$ such that for all
$M\in[0,\kappa_n^{\delta/2}]$ we have,
\begin{align}
    \label{eq:50}
    \| \psi_n \|_{L^{\infty}(\omega_{M,n})} < \frac{1}{2} \| \psi_n \|_{L^{\infty}(\omega_{M+1,n})}
  \end{align}
Then, using the {\it a priori} bound $\| \psi _n\|_{\infty} \leq 1$,
we get
\begin{align}
  \label{eq:53}
  \|\psi_n\|_{L^\infty(\omega_n)} \leq \left( \frac{1}{2}\right)^{
  \kappa_n^{\delta/2}}\,.
\end{align}
But, since $n\geq N_0$, the above bound is impossible in light of
\eqref{hc2-FK-infty-hyp2}. Therefore, \eqref{eq:52} holds for some
$M_n\in[0,\kappa_n^{\delta/2}] \cap {\mathbb N}$.

We choose now
  $P_n \in \omega_{M_n,n}$ such that
  \begin{align}
    \label{eq:54}
    |\psi_n(P_n)| =  \| \psi_n \|_{L^{\infty}(\omega_{M_n,n})},
  \end{align}
and we define
\begin{align}
  \label{eq:55}
  \Lambda_n := |\psi_n(P_n)|.
\end{align}
Then, by assumption,
\begin{itemize}
\item $\lambda_n^{-1}\,\Lambda_n \to \infty$ as $n\to\infty$.
\item $\dist(P_n, \partial \Omega) \geq \frac12\kappa_n^{-1+\delta}$.
\item $\Lambda_n \geq \frac{1}{2} \| \psi_n \|_{L^{\infty}(B(P_n, \frac{1}{2}\kappa_n^{-1+\delta/2}))}.$
\end{itemize}
Moreover, from Lemma~\ref{lem-FH-jems} we know that $\Lambda_n \to
0$ as $n\to\infty$.

Define now the following re-scaled functions, on $|x| \leq
\frac{1}{4} \kappa_n^{\delta/2}$:
\begin{align}
  \label{eq:56}
  \widetilde \varphi_n(x) &= \Lambda_n^{-1} e^{-i\sqrt{\kappa_nH_n}\,\Ab_n(P_n)} \psi_n\left(P_n+
  \frac{x}{\sqrt{\kappa_n H_n}}\right) ,\\
  \widetilde {\bf a}_n(x) &= \sqrt{\kappa_n H_n} \left(\Ab_n\left(P_n +
  \frac{x}{\sqrt{\kappa_n H_n}}\right) -\Ab_n(P_n)\right)\,.
\end{align}
Using Proposition~\ref{prop:CurlLinfty} and the assumption on
$\Lambda_n$, we know that
\begin{align}
  \label{eq:4}
  |\curl {\bf A}_n - 1| \leq \frac{C}{\kappa_n}
  (\kappa_n^{-1+\delta} + \| \psi_n \|_{L^{\infty}(\omega_n)}) \leq\frac{C}{\kappa_n}
  (\kappa_n^{-1+\delta} + \Lambda_n) \leq \frac{2C}{\kappa_n} \Lambda_n.
\end{align}
Therefore, with a new constant $C$,
\begin{align}
|\curl \widetilde {\bf a}_n - 1| = \left| (\curl {\bf A_n})\left(P_n
+
  \frac{x}{\sqrt{\kappa_n H_n}}\right) - 1\right| \leq \frac{C}{\kappa_n} \Lambda_n.
\end{align}
So we can choose a gauge function $g_n$ such that ${\bf a}_n :=
\widetilde {\bf a}_n - \nabla g_n$ satisfies
\begin{align}\label{eq:a-F}
|{\bf a}_n(x) -{\bf A}_0(x)| \leq \frac{C |x| }{ \kappa_n}
\Lambda_n\,.
\end{align}
with ${\bf A}_0$ is the magnetic potential introduced in
\eqref{eq-hc2-mpA0} corresponding  for unit constant magnetic field.

We define
\begin{align}
  \label{eq:5}
  \varphi_n := e^{-i g_n} \widetilde{\varphi}_n.
\end{align}
Using \eqref{eq:a-F} we get
\begin{align}
  \label{eq:6}
  |{\bf a}_n(x) -{\bf A}_0(x)| \leq \frac{C |x| }{ \kappa_n}
  \Lambda_n,\qquad
  |{\bf a}_n(x) +{\bf A}_0(x)| \leq |x|.
\end{align}
Furthermore, by \eqref{eq-hc2-FoHe2}, we get $|(\nabla - i {\bf
  a}_{n})\varphi_n| \leq C$, so combined with \eqref{eq:6} we get
\begin{align}
  \label{eq:7}
  |\nabla \varphi_n(x) | \leq C(1 + |x|),\\
|\varphi_n(0)| = 1,\qquad |\varphi_n(x)| \leq 2,
\end{align}
for all $|x| \leq \frac{1}{4} \kappa_n^{\delta/2}$.

The equation for $\varphi_n$ is
\begin{align}
\label{eq:phi-2} -\Delta \varphi_n - 2i {\bf a}_n \cdot \nabla
\varphi_n + |{\bf a}_n|^2 \varphi_n = \frac{\kappa_n}{H_n} (1 -
\Lambda_n^2 |\varphi_n|^2)\varphi_n,
\end{align}
on $|x| \leq \frac{1}{4} \kappa_n^{\delta/2}$. We reformulate this
as
\begin{align}
  \label{eq:57}
  [(-i\nabla + {\bf A}_0)^2-1] \varphi_n &= 2i ({\bf a}_n - {\bf
    A}_0)\nabla \varphi_n  - ({\bf
    a}_n - {\bf A}_0)({\bf a}_n + {\bf A}_0) \varphi_n \nonumber \\
&\quad+
  (\frac{\kappa_n}{H_n}-1 - \Lambda_n^2 |\varphi_n|^2) \varphi_n.
\end{align}

By elliptic estimates we get, exactly as in
\cite[(2.16)-(2.20)]{FH08}, that there exists a function
$\varphi_{\infty}\in L^{\infty}(\R^2)$ such that, up to the
extraction of a subsequence,
$$\varphi_n\stackrel{n\rightarrow\infty}{\longrightarrow}\varphi_\infty$$
holds in $C^1(K)$ on any compact subset $K\subset\R^2$. Moreover,
$\varphi_\infty$ satisfies
 $|\varphi_{\infty}(0)|=1$ and
\begin{align}
  \label{eq:59}
  [(-i\nabla + {\bf A}_0)^2-1] \varphi_{\infty} = 0\quad{\rm in}~\R^ 2\,.
\end{align}

Consider now a localization function $\chi \in C^{\infty}(\R)$ with
$\chi(t)=1$ for $t \leq 1/2$, $\chi(t) = 0$ for $t\geq 1$, and
define $\chi_n(x) = \chi(\kappa_n^{-\eta} |x|)$ for some
$\eta\in(0,\frac\delta2)$. We have the following equation for
$\chi_n \varphi_n$:
\begin{align}
  \label{eq:60}
  [(-i\nabla + {\bf A}_0)^2&-1] (\chi_n\varphi_n) \nonumber \\
&= \chi_n [(-i\nabla +
  {\bf A}_0)^2-1] \varphi_n -2i (\nabla \chi_n) \cdot (-i\nabla +
  {\bf A}_0) \varphi_n - (\Delta \chi_n) \varphi_n \nonumber \\
  &=2i ({\bf a}_n - {\bf
    A}_0)\chi_n \nabla \varphi_n  - ({\bf
    a}_n - {\bf A}_0)({\bf a}_n + {\bf A}_0) \chi_n\varphi_n \nonumber \\
&\quad+
  (\frac{\kappa_n}{H_n}-1 - \Lambda_n^2 |\varphi_n|^2) \chi_n \varphi_n
  -2i (\nabla \chi_n) \cdot (-i\nabla +
  {\bf A}_0) \varphi_n - (\Delta \chi_n) \varphi_n\,.
\end{align}

We introduce the projector $\Pi_0$ on the lowest Landau level. This
projector is given explicitly by the integral kernel,
$$\Pi_0(x,y)=\frac1{2\pi}e^{\frac{i}2(x_1y_2-x_2y_1)}e^{-\frac12(x-y)^2}\,,$$
and is continuous on $L^p(\R^2)$ for all $p\in[2,\infty]$.

Fix a function $f \in C_0^{\infty}(\R^2)$. We will prove that
\begin{align}\label{eq:weak}
\int_{\R^2} \big(\Pi_0 f(x)\big) |\varphi_{\infty}(x)|^2
\varphi_{\infty}(x)\,dx = 0.
\end{align}
Since $f$ is arbitrary, this implies that
\begin{align}
\label{eq:complementary} \Pi_0\left(
|\varphi_{\infty}|^2\varphi_{\infty}\right) =0\quad{\rm in}~\R^2.
\end{align}
But a result from \cite{FH08} says that \eqref{eq:complementary}
combined with \eqref{eq:59} implies that $\varphi_{\infty} \equiv
0$. This is in contradiction to the fact that
$|\varphi_{\infty}(0)|=1$. Therefore we have reached a
contradiction.

Thus, it only remains to prove \eqref{eq:weak}.

By definition of $\Pi_0$ we have
\begin{align}
\int_{\R^2} \big(\Pi_0 f(x)\big)  [(-i\nabla + {\bf A}_0)^2-1]
(\chi_n\varphi_n) \,dx= 0.
\end{align}
By consequence,
\begin{align}\label{eq:limit}
  \lim_{n\rightarrow \infty} \Lambda_n^{-2} \int_{\R^2} \big(\Pi_0 f(x)\big)  [(-i\nabla + {\bf A}_0)^2-1]
   (\chi_n\varphi_n) \,dx= 0.
\end{align}
We insert \eqref{eq:60} in \eqref{eq:limit}. Consider first the term
with $(\frac{\kappa_n}{H_n}-1 - \Lambda_n^2 |\varphi_n|^2) \chi_n
\varphi_n$. By assumption
$$
\Lambda_n^{-2} |\frac{\kappa_n}{H_n} -1 | \leq \Lambda_n^{-2}
\lambda_n^2 \to 0\quad{\rm as~}n\to\infty\,.
$$
So one readily gets the convergence,
\begin{align}
 \Lambda_n^{-2} \int_{\R^2} \big(\Pi_0 f(x)\big) (\frac{\kappa_n}{H_n}-1 - \Lambda_n^2 |\varphi_n|^2) \chi_n \varphi_n\,dx &
 \stackrel{n \rightarrow \infty}{\longrightarrow}
- \int_{\R^2} \big(\Pi_0 f(x)\big) |\varphi_{\infty}|^2
\varphi_{\infty}\,dx.
\end{align}
So in order to obtain \eqref{eq:weak} we only have to prove that the
other terms from \eqref{eq:60} vanish in the limit.

By \eqref{eq:6} and \eqref{eq:7}
\begin{align}
 \Big| \Lambda_n^{-2} \int_{\R^2} \big(\Pi_0 f(x)\big) 2i ({\bf a}_n - {\bf
    A}_0)\chi_n \nabla \varphi_n  &- (({\bf
    a}_n - {\bf A}_0)({\bf a}_n + {\bf A}_0) \chi_n\varphi_n \,dx\Big| \nonumber \\
    &\leq
     \Lambda_n^{-2} \int_{\R^2} \big|\Pi_0 f(x)\big| C  \frac{(1+|x|^2)\Lambda_n}{\kappa_n} \,dx \nonumber \\
     &\to0\quad{\rm as}~n\to\infty\,,
\end{align}
where we used the assumption that $\Lambda_n \gg \kappa_n^{-1}$ as
$n\to\infty$.

Notice, using the compact support of $f$ and the off-diagonal decay
of $\Pi_0$, that $\Pi_0 f(x)$ is exponentially small on $\{ |x| \geq
\kappa^{-\eta}\}$. Therefore, it is easy to see that also
\begin{align}
 \Lambda_n^{-2} \int_{\R^2} \big(\Pi_0 f(x)\big) \{2i (\nabla \chi_n) \cdot (-i\nabla +
  {\bf A}_0) \varphi_n +(\Delta \chi_n) \varphi_n\} \,dx \to 0
\quad{\rm as~}n\to\infty.
\end{align}
This finishes the proof of \eqref{eq:weak} and therefore of
Theorem~\ref{thm:Linfty-bulk}.
\end{proof}

\section{Upper bound of the energy}\label{hc2-sec-ub}
In this section we construct test configurations and compute their
energies, obtaining thus upper bounds for the functional
$\mathcal E$ in (\ref{eq-hc2-GL}).
Recall the definition of the ground state energy
$C_0(\kappa,H)$ in \eqref{eq-GL-GS}. We will prove the following theorem.

\begin{thm}\label{thm-hc2-FK-ub}
Assume that the magnetic field satisfies
$$H=\kappa-\mu(\kappa)\sqrt{\kappa}\,,$$
with  $\mu:\R_+\mapsto\R$ a function such that
$$\lim_{\kappa\to\infty}\frac{\mu(\kappa)}{\sqrt{\kappa}}=0\,.$$

Then, as $\kappa\to\infty$, the following upper bound  holds
for the ground state energy introduced in  (\ref{eq-GL-GS}),
\begin{equation}\label{eq-GS-UB}
C_0(\kappa,H)\leq
-E_1|\partial\Omega|\kappa
-E_2|\Omega|\,[\mu(\kappa)]_+^2\kappa+o(\max(1,[\mu(\kappa)]_+^2)\kappa)\,.
\end{equation}
Here $E_1>0$ and $E_2>0$ are the constants
introduced in (\ref{eq-hc2-E1}) and (\ref{eq-hc2-E2}) respectively.
\end{thm}

\subsection{Boundary configuration}

\subsubsection{Boundary coordinates}
In order to treat the surface (boundary) energy contribution, we
shall frequently pass to a coordinate system valid in a tubular
neighborhood of $\partial\Omega$.  For more details on these
coordinates, see for instance \cite[Appendix F]{FoHe08}.

For a sufficiently small $t_0>0$, we introduce the open set
$$\Omega(t_0)=\{x\in\mathbb R^2~:~{\rm
  dist}(x,\partial\Omega)<t_0\}.$$
Let $s\mapsto\gamma(s)$ be the  parametrization of $\partial\Omega$
by arc-length and $\nu(s)$ the unit inward normal of
$\partial\Omega$ at
  $\gamma(s)$.

Define the transformation
\begin{align}
  \label{eq:19}
\Phi:
\left[-\frac{|\partial\Omega|}2,\frac{|\partial\Omega|}2\right[\,\times
]-t_0,t_0[\ni(s,t)\mapsto \gamma(s)+t\nu(s)\in \Omega(t_0)
\end{align}
and extend it to $\R\times]-t_0,t_0[$ by periodicity with respect to
$s$. The resulting transformation
becomes a  local  diffeomorphism whose Jacobian is $|D\Phi|=1-tk
 (s)$, where $k$ denotes the curvature of $\partial\Omega$. For $x\in\Omega(t_0)$, we put
$$\Phi^{-1}(x)=(s(x),t(x))$$
and we get in particular that \begin{equation} \label{eq-hc2-t(x)}
t(x)={\rm dist}(x,\partial\Omega)\,.\end{equation}

Using the coordinate transformation $\Phi$, we associate to any
function $u\in L^2(\Omega)$, a function $\widetilde u$ defined in $
[-\frac{|\partial\Omega|}2,
\frac{|\partial\Omega|}2[\,\times\,[0,t_0]$ by,
\begin{equation}\label{VI-utilde}
\widetilde u(s,t)=u(\Phi^{-1}(s,t))\,,
\end{equation}
and we will use the symbol $U_{\Phi}$ for the operator that maps $u$
to $\widetilde u$. Notice also that the function $\widetilde u$
extends naturally to a $|\partial\Omega|$-periodic function in
$s\in\mathbb R$.

We get then the following change of variable formulae.
\begin{proposition}\label{hc2-App:transf}
Let $u\in H^1(\Omega(t_0))$ and $\Ab\in H^1(\Omega;\R^2)$. We write
$\widetilde u(s,t)=u(\Phi(s,t))$,
\begin{align*}
\widetilde \Ab_1 = \Ab_1 \circ \Phi,\qquad \widetilde \Ab_2 = \Ab_2 \circ
\Phi\,,\quad g(s,t)=1-tk(s)\,.
\end{align*}
Then we have~:
\begin{multline}\label{App:qfstco}
\int_{\Omega(t_0)}\left|(\nabla-i\Ab)u\right|^2dx =
\int_{-\frac{|\partial\Omega|}2}^{\frac{|\partial\Omega|}2}
\int_{0}^{t_0}\left[ [g(s,t)]^{-2}|(\partial_s-i\widetilde
\Ab_1)\widetilde u|^2+|(\partial_t-i\widetilde \Ab_2)\widetilde
u|^2\right]g(s,t)\,dtds,
\end{multline}
and
\begin{equation}\label{App:nostco}
\int_{\Omega(t_0)}
|u(x)|^2\,dx=\int_{-\frac{|\partial\Omega|}2}^{\frac{|\partial\Omega|}2}
\int_{0}^{t_0} |\widetilde u(s,t)|^2g(s,t)\,dtds.
\end{equation}
\end{proposition}

Recall the vector field $\FF$ introduced in \eqref{eq-hc2-F}.
Another feature of the coordinate system $(s,t)$ is that it permits
us to express $\FF$ in a more explicit form (up to a gauge
transformation).

Let us introduce the following two subsets of $\Omega(t_0)$,
\begin{equation}\label{eq-hc2-U12}
U_1=\Phi^{-1}\left(\left[-\frac{|\partial\Omega|}2,0\right)\times[0,t_0)\right)\,,\quad
U_2=\Phi^{-1}\left(\left[0,\frac{|\partial\Omega|}2\right)\times[0,t_0)\right)\,.
\end{equation}

\begin{lem}\label{eq-hc2-gaugeT}
There exist two functions $\chi_1\in C^2(U_1;\R)$ and $\chi_2\in
C^2(U_2;\R)$ such that upon setting
$$\FF^1=\FF+\nabla\chi_1\quad{\rm in~}U_1\,,\quad \FF^2=\FF+\nabla\chi_2\quad{\rm in~}U_2\,,$$
then we have in the $(s,t)$ coordinates,
$$\left(U_\Phi\FF^j\right)(s,t)=\left(-t+k(s)\frac{t^2}2,0\right)\quad{\rm
in}~\Phi^{-1}(U_j)\,,\quad j=1,2\,.
$$
\end{lem}

The  gauge
transformation in Lemma~\ref{eq-hc2-gaugeT} can not be applied
 globally in $\Omega(t_0)$, or otherwise one has to add a geometric
constant $\gamma_0$ in the expression of the obtained field, see
\cite[Lemma~F.1.1]{FoHe08}. In order to avoid the presence of such a
geometric constant, we partition the domain into two different subsets
and work seperately in each of them.

\subsubsection{The test configuration}

Let us introduce, for reasons of convenience  that will become
clear, the following small parameter,
\begin{equation}\label{eq-hc2-epsilon}
\varepsilon=\frac1{\sqrt{\kappa H}}\,.
\end{equation}
Let, for $\ell>0$, $\phi_\ell$ be
a minimizer of
the reduced Ginzburg-Landau energy $\mathcal E_\ell$ from
(\ref{eq-hc2-redGL}), see Theorem~\ref{thm-hc2-Pa02}.
We make the following choice of $\ell$,
\begin{equation}\label{eq-hc2-ell}
\ell=\frac{|\partial\Omega|}{4\varepsilon}\,,
\end{equation}
Define, for  $(s,t)\in [-\frac{|\partial\Omega|}2,\frac{|\partial\Omega|}2)$,

\begin{equation}\label{eq-hc2-tf:bnd}
\varphi_{\rho,\varepsilon}(s,t)= \left\{
\begin{array}{lll}
\chi\left(\displaystyle\frac{t}{\varepsilon^\rho}\right)\phi_\ell\left(\displaystyle
\frac{s}\varepsilon+\ell,\displaystyle\frac{t}\varepsilon\right)\,,&{\rm
if}&-\displaystyle\frac{|\partial\Omega|}2\leq s<0\,,\\
\chi\left(\displaystyle\frac{t}{\varepsilon^\rho}\right)\phi_\ell\left(\displaystyle
\frac{s}\varepsilon-\ell,\displaystyle\frac{t}\varepsilon\right)\,,&{\rm
if}&0\leq s<\displaystyle\frac{|\partial\Omega|}2\,.
\end{array}\right.
\end{equation}
 The parameter $\rho\in(0,1)$ is to be chosen
later, and $\chi\in C_0^\infty(\R)$ is a standard cut-off function
satisfying,
$$0\leq\chi\leq 1~{\rm in}~\R\,,\quad \chi=1~{\rm in}~
\big{(}-\frac12,\frac12\big{)}\,,\quad{\rm and}~{\rm
supp}\,\chi\subset[-1,1]\,.
$$
The function $\varphi_{\rho,\varepsilon}$ is clearly in $H^1$, since
$\phi_\ell$ vanishes for $s=\pm \ell$. Using the coordinate
transformation \eqref{eq:19}, we get from
$\varphi_{\rho,\varepsilon}$ a test function in $\Omega$\,,
\begin{equation}\label{eq-hc2-tf:bnd'}
\psi_{\rho,\varepsilon}^{\rm bnd}(x)=\left(\mathbf
1_{U_1}(x)\,e^{-i\kappa H\chi_1(x)}+\mathbf 1_{U_2}(x)
\,e^{-i\kappa H\chi_2(x)}\right)\varphi_{\rho,\varepsilon}(\Phi^{-1}(x))\,,\quad
\forall~x\in\Omega(t_0)\,,
\end{equation}
and extended by $0$ on $\Omega\setminus\Omega(t_0)$. Here the gauges
$\chi_1$ and $\chi_2$ are introduced in Lemma~\ref{eq-hc2-gaugeT}
above.

\begin{lem}\label{lem-hc2-ub:bnd}
Given two positive constants $m$ and $M$ with $m<M$, there exist
positive constants $\varepsilon_0$ and $C$ such that if the magnetic
field satisfies $m\kappa\leq H\leq M\kappa$, then for all
$\varepsilon\in(0,\varepsilon_0]$, $\rho\in(0,1)$ and
$\delta\in(0,1)$, the following estimate holds,
$$\mathcal E(\psi_{\rho,\varepsilon}^{\rm bnd},\FF;\Omega)\leq
2d(\ell) +C\left(\varepsilon^{4-5\rho}+\varepsilon^{\rho/2}
+\varepsilon^{1-\rho}+\left|\frac{\kappa}{H}-1\right|\right)\ell\,.$$
Here $\ell=|\partial\Omega|/(4\varepsilon)$, $d(\ell)$ is introduced
in (\ref{eq-hc2-d(ell)}), $\mathcal E$, $\FF$ and
$\psi_{\rho,\varepsilon}^{\rm bnd}$ are the energy functional, the
vector field and the test configuration introduced in
(\ref{eq-hc2-GL}), (\ref{eq-hc2-F}) and (\ref{eq-hc2-tf:bnd'})
respectively.
\end{lem}
\begin{proof}
Let us introduce the vector field
$$A(s,t)=(A_1(s,t),A_2(s,t))=
\left(-t+k(s)\frac{t^2}2,0\right)\,,\quad\forall~(s,t)\in\R^2\,,$$
together with the two energies,
$$A_\varepsilon=\int_{-\frac{|\partial\Omega|}2}^0\int_0^{t_0}\Big(
[g(s,t)]^{-2}|(\partial_s-i\varepsilon^{-2} A_1)
\phi_{\rho,\varepsilon}|^2
+|\partial_t\phi_{\rho,\varepsilon}|^2-\kappa^2|\phi_{\rho,\varepsilon}|^2
+\frac{\kappa^2}2|\phi_{\rho,\varepsilon}|^4\Big)g(s,t)\,dtds\,,
$$
and
$$
B_\varepsilon=\int_{0}^{-\frac{|\partial\Omega|}2}\int_0^{t_0}\Big(
[g(s,t)]^{-2}|(\partial_s-i\varepsilon^{-2}A_1)
\phi_{\rho,\varepsilon}|^2
+|\partial_t\phi_{\rho,\varepsilon}|^2-\kappa^2|\phi_{\rho,\varepsilon}|^2
+\frac{\kappa^2}2|\phi_{\rho,\varepsilon}|^4 \Big)g(s,t)\,dtds\,.
$$
Here $g(s,t)=1-tk(s)$  and $t_0$ a sufficiently small constant as
previously. Proposition~\ref{hc2-App:transf},
Lemma~\ref{eq-hc2-gaugeT} and the definition of the function
$\psi^{\rm bnd}_{\rho,\varepsilon}$ all together give,
$$\mathcal E(\psi^{\rm bnd}_{\rho,\varepsilon},\mathbf
F;\Omega)=A_\varepsilon+B_\varepsilon\,.
$$
We claim that
$$A_\varepsilon\leq d(\ell)+C\left(\varepsilon^{4-5\rho}+\varepsilon^{\rho/2}
+\varepsilon^{1-\rho}+\left|\frac{\kappa}{H}-1\right|\right)\ell\,.$$
and $$B_\varepsilon\leq
d(\ell)+C\left(\varepsilon^{4-5\rho}+\varepsilon^{\rho/2}
+\varepsilon^{1-\rho}+\left|\frac{\kappa}{H}-1\right|\right)\ell\,.$$
Let us prove the upper bound for $A_\varepsilon$. The upper bound
for $B_\varepsilon$ follows in the same way as for $A_\varepsilon$.

Define the rescaled variables $\sigma=\frac{s}{\varepsilon}+\ell$,
$\tau=\frac{t}{\varepsilon}$ and the rescaled function
$$u(\sigma,\tau)=\phi_{\rho,\varepsilon}(s,t)\,,\quad\forall~(\sigma,\tau)\in(-\ell,\ell)\times(0,\varepsilon^{\rho-1})\,,
$$ which is extended by zero for $(\sigma,\tau)\in(-\ell,\ell)\times[\varepsilon^{\rho-1},\infty)$.

In the new scale, the expression for $A_\varepsilon$ becomes,
\begin{equation}\label{a-epsilon}
A_\varepsilon=\int_{-\ell}^{\ell}\int_0^{\varepsilon^{\rho-1}}
\left(
[g_\varepsilon(\sigma,\tau)]^{-2}|(\partial_\sigma+i\tau-i\varepsilon
a_\varepsilon ) u|^2 +|\partial_\tau u|^2-\frac\kappa{H}|u|^2
+\frac{\kappa}{2H}|u|^4 \right)g_\varepsilon(\sigma,\tau)\,d\tau
d\sigma\,.
\end{equation}
Here
$$g_\varepsilon(\sigma,\tau)=1-\varepsilon k_\varepsilon(\sigma)
\tau\,,\quad
k_\varepsilon(\sigma)=k(\varepsilon(\sigma-\ell))\,,\quad
a_\varepsilon(\sigma,\tau)=k_\varepsilon(\sigma)\frac{\tau^2}2\,.$$
There exists a positive constant $C$ such that, for $\varepsilon$
sufficiently small, the following estimate holds,
\begin{equation}\label{eq-hc2-ub-E1}
\frac12<1-C\varepsilon^\rho\leq g_\varepsilon(\sigma,\tau)\leq
1+C\varepsilon^\rho\,,\quad\forall~(\sigma,\tau)\in(-\ell,\ell)\times(0,\varepsilon^{\rho-1})\,.\end{equation}
Replacing $C$ by a larger constant, it holds for $\delta\in(0,1)$,
\begin{align*}
\int_{-\ell}^\ell\int_0^{\varepsilon^{\rho-1}}|\partial_\tau
u|^2d\tau
d\sigma&=\int_{-\ell}^\ell\int_0^{\varepsilon^{\rho-1}}|\chi_\varepsilon(\tau)\partial_\tau\phi_\ell+\phi_\ell\partial_\tau
\chi_\varepsilon(\tau)|^2d\tau d\sigma\\
&\leq
(1+\delta)\int_{-\ell}^\ell\int_0^{\infty}|\partial_\tau\phi_\ell|^2d\tau
d\sigma+C\delta^{-1}\varepsilon^{2-2\rho}
\int_{-\ell}^\ell\int_{\frac12\varepsilon^{\rho-1}}^{\varepsilon^{\rho-1}}|\phi_\ell|^2d\tau
d\sigma\,.\end{align*} Here
$\chi_\varepsilon(\tau)=\chi(\frac{\varepsilon
\tau}{\varepsilon^{\rho}})$ and $\phi_\ell$ a minimizer of the
functional $\mathcal E_\ell$ in \eqref{eq-hc2-redGL}.

Using the decay of $\phi_\ell$ in Theorem~\ref{thm-hc2-Pa02}, we
get, \begin{equation}\label{eq-hc2-ub-E2}
\int_{-\ell}^\ell\int_0^{\varepsilon^{\rho-1}}|\partial_\tau
u|^2d\tau
d\sigma\leq(1+\delta)\int_{-\ell}^\ell\int_0^{\infty}|\partial_\tau\phi_\ell|^2d\tau
d\sigma+C\delta^{-1}\varepsilon^{4-4\rho}|\ln\varepsilon|\ell\,.\end{equation}
In a similar fashion we get,
\begin{align}\label{eq-hc2-ub-E3}
&\int_{-\ell}^\ell\int_0^{\varepsilon^{\rho-1}}|(\partial_\sigma+i\tau-i\varepsilon
a_\varepsilon ) u|^2d\tau d\sigma\nonumber\\
&\hskip2cm\leq
\int_{-\ell}^\ell\int_0^{\varepsilon^{\rho-1}}\left((1+\delta)
|(\partial_\sigma+i\tau)\phi_\ell|^2+C\delta^{-1}\varepsilon^2|
a_\varepsilon  u|^2\right)d\tau d\sigma\nonumber\\
&\hskip2cm\leq(1+\delta)\int_{-\ell}^\ell\int_0^{\varepsilon^{\rho-1}}
|(\partial_\sigma+i\tau)\phi_\ell|^2d\tau
d\sigma+C\delta^{-1}\varepsilon^{2\rho}|\ln\varepsilon|\ell\,.
\end{align}
Writing,
\begin{align*}
\frac{\kappa}{H}\int_{-\ell}^\ell\int_0^{\varepsilon^{\rho-1}}|u|^2d\tau
d\sigma&=\int_{-\ell}^\ell\int_0^{\infty}|\phi_\ell|^2d\tau
d\sigma+\int_{-\ell}^\ell\int_{\frac12\varepsilon^{\rho-1}}^{\infty}(|\chi_\varepsilon|^2-1)|\phi_\ell|^2d\tau
d\sigma\\
&+\left(\frac{\kappa}{H}-1\right)\int_{-\ell}^\ell\int_0^{\varepsilon^{\rho-1}}|u|^2d\tau
d\sigma\,,
\end{align*}
we get by the decay of $\phi_\ell$,
\begin{equation}\label{eq-hc2-ub-E4}
\frac{\kappa}{H}\int_{-\ell}^\ell\int_0^{\varepsilon^{\rho-1}}|u|^2d\tau
d\sigma\geq\int_{-\ell}^\ell\int_0^{\infty}|\phi_\ell|^2d\tau
d\sigma-C\varepsilon^{1-\rho}\ell-C\left|\frac\kappa{H}-1\right|\ell\,.
\end{equation}
Similarly, we have the upper bound,
\begin{equation}\label{eq-hc2-ub-E5}
\frac{\kappa}{H}\int_{-\ell}^\ell\int_0^{\varepsilon^{\rho-1}}|u|^4d\tau
d\sigma\leq\int_{-\ell}^\ell\int_0^{\infty}|\phi_\ell|^2d\tau
d\sigma+C\left|\frac\kappa{H}-1\right|\ell\,.
\end{equation}
Inserting the estimates \eqref{eq-hc2-ub-E1}-\eqref{eq-hc2-ub-E5}
into \eqref{a-epsilon} we get,
$$A_\varepsilon\leq \mathcal E_\ell(\phi_\ell)
+C\left(\delta^{-1}|\ln\varepsilon|(\varepsilon^{2\rho}+\varepsilon^{4-4\rho})+\delta+\varepsilon^\rho
+\varepsilon^{1-\rho}+\left|\frac{\kappa}{H}-1\right|\right)\ell\,.$$
Here $\mathcal E_\ell$ is the functional in \eqref{eq-hc2-redGL}.
Remembering that $\phi_\ell$ is a minimizer of $\mathcal E_\ell$,
the definition of $d(\ell)$, and choosing
$\delta=\varepsilon^{\rho/2}$, we get the desired upper bound on
$A_\varepsilon$.
\end{proof}

\subsection{Bulk configuration}
In this section we construct a test configuration $(\psi^{\rm
int},\Ab)$ using the limiting problem (\ref{eq-hc2-eneAb}).  Let us
take $R>1$ (that will be chosen as a function of $\varepsilon$ such
that $\varepsilon R\to0$ as $\varepsilon\to0$).
Then, thanks to Proposition~\ref{prop-hc2-exmin}, the functional
$F_{R}$ in (\ref{eq-hc2-eneAb}) admits a minimizer $f_{R}$ in
$L_{R}$, and we denote by,
$$c(R)=F_{R}(f_{R}).$$
Recall the magnetic potential $\Ab_0$ introduced in
(\ref{eq-hc2-mpA0}). The configuration $(f_{R}, \Ab_0)$,
defined initially on the unit lattice $K_{R}$, can be defined by
periodicity in all $\R^2$.

Let us define now the following test configuration in $\Omega$,
\begin{equation}\label{eq-hc2-tf:int}
\psi_{\rho,R,\varepsilon}^{\rm
int}(x)=\kappa^{-1/4}\,h\left(R^{\rho}\frac{\dist(x,\partial\Omega)}{2}\right)\,
f_{R}\left(\frac{x}{\varepsilon}\right)\,,\quad\forall~x\in\Omega\,.
\end{equation}
Here $\rho>0$ is to be fixed later, and $h\in C^\infty(\R)$ is a
cut-off function such that
$$0\leq h\leq 1{\rm ~in~}\R\,,\quad h=1{~\rm
in~}[1,\infty)\,,\quad{\rm supp}\,h\subset(0,\infty)\,.$$ Notice
that we can cover $\Omega$ by $N_\varepsilon$ squares of the lattice
$\varepsilon R(\mathbb Z\oplus i\mathbb Z)$, where $N_\varepsilon$
satisfies,
\begin{equation}\label{eq-hc2-N(ep)}
\lim_{\varepsilon\to 0}\left[N_\varepsilon \times(\varepsilon^2R^2
)\right]=|\Omega|\,.
\end{equation}

The next lemma gives an estimate of the energy of the test
configuration \eqref{eq-hc2-tf:int}.

\begin{lem}\label{lem-hc2-up:int}
Let $\rho\in(0,1)$ be a given constant.  Assume that $R=R(\varepsilon)$
is a function satisfying  $|K_{R}|\in2\pi\mathbb N$ and $1\ll R\ll
\varepsilon^{-\frac1{1+\rho}}$ as $\varepsilon\to0$. There exist positive
constants $\varepsilon_0$, $C$, and a function
$$\R_+\ni t\mapsto\delta(t)\in\R_+\,,\quad \lim_{t\to\infty}\delta(t)=0\,,$$ such that, for all
$\varepsilon\in(0,\varepsilon_0]$,  $\mu>0$, and if the
magnetic field satisfies,
$$H=\kappa-\mu\sqrt{\kappa}\,,$$ then we have
the following estimate,
$$\mathcal E\left(\mu^{1/2}\psi_{\rho,R,\varepsilon}^{\rm
int},\Ab_0 \right)\leq \mu^2\kappa|\Omega|\, c(R)
+C\mu\left(\kappa^{-1/2}R^\rho+\kappa^{3/2}R^{-\rho}\right)
+\mu^2\kappa\delta(\kappa)\,.$$ Here $c(R)$ is introduced in
\eqref{eq-hc2-c(r,t)}.
\end{lem}
\begin{proof}
Let us denote by $\widetilde f_{R}(x)=f_{R}(x/\varepsilon)$ and
$h_R(x)=h(R^\rho \frac{{\rm dist}(x,\partial\Omega)}2)$. Notice that we have the following
localization formula,
\begin{multline}\label{eq-hc2-ims}
\int_{\Omega}|(\nabla-i\varepsilon^{-2}\Ab_0
)\psi_{\rho,R,\varepsilon}^{\rm int}|^2\,dx=\kappa^{-1/2}\,\Re
\langle -h^2_R(\nabla-i\varepsilon^{-2}\Ab_0 )^2\widetilde f_{R}\,,\,\widetilde
f_{R}\rangle_{L^2(\Omega)}\\+\kappa^{-1/2}\int_{\Omega}|\,|\nabla
h_R|^2\,\widetilde f_{R}|^2\,dx\,.
\end{multline}
From the definition of $\widetilde f_{R}$ and $\Ab_0$, and  a simple
scaling, we may check that
$$-(\nabla-i\varepsilon^{-2}\Ab_0 )^2\widetilde
f_{R}=\varepsilon^{-2} \widetilde f_{R}=\kappa H\widetilde f_{R}\,.
$$
Therefore, \eqref{eq-hc2-ims} becomes,
\begin{align}
&\hskip-1cm\kappa^{1/2}\int_{\Omega}|(\nabla-i\varepsilon^{-2}\Ab_0)\psi_{\rho,R,\varepsilon}^{\rm int}|^2\,dx\nonumber \\
&= \kappa H \int_\Omega|h_R \widetilde
f_{R}|^2\,dx+\int_\Omega|\nabla h_R|^2\,|\widetilde
f_{R}|^2\,dx\nonumber\\
&\leq \kappa H\int_\Omega|\widetilde
f_{R}|^2\,dx+2\|h'\|_{L^\infty(\R)}^2R^{2\rho}\int_{\{2R^{-\rho}\leq\dist(x,\partial\Omega)\leq
4R^{-\rho}\}}|\widetilde f_{R}|^2\,dx\,.\label{eq-hc2-ims'}
\end{align}
Let us estimate the last term in \eqref{eq-hc2-ims'}, which is in
fact a remainder term. Recall that $\widetilde f_{R}$ is periodic
with respect to the lattice $\varepsilon R(\mathbb Z\oplus i\mathbb
Z)$. Using the condition $R\ll\varepsilon^{-\frac1{1+\rho}}$, we
cover $\{x\in\Omega~:~ 2R^{-\rho}\leq\dist(x,\partial\Omega)\leq
4R^{-\rho}\}$ by $N'_\varepsilon$ squares of the lattice
$\varepsilon R(\mathbb Z\oplus i\mathbb Z)$, with
$N_{\varepsilon}'\leq C\frac{R^{-\rho}}{\varepsilon^2R^2 }$.
Therefore, 
$$
\int_{\{2R^{-\rho}\leq\dist(x,\partial\Omega)\leq
4R^{-\rho}\}}|\widetilde f_{R}|^2\,dx\leq
C\frac{R^{-\rho}}{\varepsilon^2R^2 }\int_{K_{\varepsilon
R}}|\widetilde f_{R}(x)|^2\,dx=CR^{-2-\rho}\int_{K_{
R}}|f_{R}(x)|^2\,dx\,.
$$
Invoking the estimate of
Proposition~\ref{prop-hc2-exmin}, we get,
\begin{equation}\label{eq-hc2-ims:error}
\int_{\{2R^{-\rho}\leq\dist(x,\partial\Omega)\leq
4R^{-\rho}\}}|\widetilde f_{R}|^2\,dx\leq C R^{-\rho}\,.
\end{equation}
Next we estimate $\|h_R\widetilde f_{R}\|_{L^2(\Omega)}$ from below.
Notice that, since $0\leq h_R\leq 1$ and $1-h_R$ is supported in a
thin neighborhood near the boundary, we have,
\begin{align*}
\int_\Omega|h_R(x)\widetilde f_{R}(x)|^2\,dx&=
\int_\Omega|\widetilde
f_{R}(x)|^2\,dx-\int_\Omega(1-h_R^2(x))|\widetilde
f_{R}(x)|^2\,dx\\
&\geq \int_\Omega|\widetilde
f_{R}(x)|^2\,dx-\int_{\{x\in\Omega~:~\dist(x,\partial\Omega)\leq
2R^{-\rho}\}}|\widetilde f_{R}(x)|^2\,dx\,.
\end{align*}
Similarly as for \eqref{eq-hc2-ims:error}, the estimate of
Proposition~\ref{prop-hc2-exmin} gives
$$\int_{\{x\in\Omega~:~\dist(x,\partial\Omega)\leq
2R^{-\rho}\}}|\widetilde f_{R}(x)|^2\,dx\leq C R^{-\rho}\,,$$ and
therefore,
\begin{equation}\label{eq-hc2-ims:error'}
\int_\Omega|h_R(x)\widetilde f_{R}(x)|^2\,dx\geq
\int_\Omega|\widetilde f_{R}(x)|^2\,dx -C R^{-\rho}\,.
\end{equation}
Collecting \eqref{eq-hc2-ims'},
\eqref{eq-hc2-ims:error} and \eqref{eq-hc2-ims:error'}, and
remembering that $\curl \Ab_0 =1$ by construction, we get finally,
\begin{align}\label{eq-hc2-ims:collection}
\mathcal E(\mu^{1/2}\psi_{\rho,R,\varepsilon}^{\rm int},\Ab_0) &\leq
\sqrt{\kappa}\,( H-\kappa)\mu\int_{\Omega}|\widetilde
f_{R}|^2\,dx+\frac{\kappa \mu^2}2\int_\Omega|\widetilde
f_{R}|^4\,dx+C\mu(\kappa^{-1/2}R^{\rho}+\kappa^{3/2}R^{-\rho})
\nonumber\\
&=\mu^2\kappa\int_\Omega\left(\frac12|\widetilde
f_{R}|^4-|\widetilde f_{R}|^2\right)\,dx+
C\mu\left(\kappa^{-1/2}R^\rho+\kappa^{3/2}R^{-\rho}\right)\,.
\end{align}
We have to estimate the integral in (\ref{eq-hc2-ims:collection}).
Toward that end, we define two sets $\underline{\mathcal I}$ and
$\overline{\mathcal I}$ as follows. A square $K$ of the lattice
$\varepsilon R(\mathbb Z\oplus i\mathbb Z)$ belongs to
$\underline{\mathcal I}$ if $K\subset\Omega$; if
$K\cap\Omega\not=\emptyset$ then $K$ belongs to $\overline{\mathcal
  I}$. Let us introduce the two integers,
$$
\underline N_\varepsilon={\rm Card}( \underline{\mathcal I})\,,\quad
\overline
N_\varepsilon={\rm Card}(\overline{\mathcal  I})\,.
$$
The formula in \eqref{eq-hc2-N(ep)} still holds for both $\underline
N_\varepsilon$ and $\overline N_\varepsilon$. Furthermore, by
periodicity of $|\widetilde f_R|$, we get,
\begin{align*}
\int_\Omega\frac12|\widetilde f_{R}|^4-|\widetilde f_{R}|^2& \leq
\overline N_\varepsilon\frac12\int_{K_{\varepsilon R}}|\widetilde
f_{R}|^4-
\underline N_\varepsilon\int_{K_{\varepsilon R}}|\widetilde f_{R}|^2\\
&=\frac{|\Omega|}{|K_{R}|}\left(\frac{1+o(1)}2\int_{K_{R}}|f_{R}|^4
-(1+o(1))\int_{K_{R}}|f_{R}|^2\right)\\
&=|\Omega|c(R)+o(1)\quad{\rm as }~\varepsilon\to0\,.
\end{align*}
In the last step above we used the definition of $f_R$ and
Proposition~\ref{prop-hc2-exmin}.

Upon substitution in \eqref{eq-hc2-ims:collection}, we get,
\begin{align*}
&\hskip-1cm \mathcal E(\mu^{1/2}\psi_{\rho,R,\varepsilon}^{\rm
int},\Ab_0) \\
&\leq \mu^2\kappa\big{(}|\Omega|\,c(R) 1+o(1)\big{)}+
C\mu\left(\kappa^{-1/2}R^\rho+\kappa^{3/2}R^{-\rho}\right)\quad{\rm
as}~\varepsilon\to0\,,
\end{align*}
which is what we wanted to prove.
\end{proof}

\subsection{Proof of Theorem~\ref{thm-hc2-FK-ub}}
Let $\mu=\mu(\kappa)$ be given by
$H=\kappa-\mu(\kappa)\sqrt{\kappa}$. Let us define the following
test function,
$$\psi(x)=\psi_{\rho,\varepsilon}^{\rm
bnd}(x)+[\mu]^{1/2}_+e^{-i\kappa
H\varphi_0}\psi_{\rho,R,\varepsilon}^{\rm bnd}(x)\,,
$$
and evaluate the energy $\mathcal E(\psi,\FF;\Omega)$. Here
$\psi^{\rm bnd}_{\rho,\varepsilon}$ is introduced in
(\ref{eq-hc2-tf:bnd'}), $\psi_{\rho,R,\varepsilon}^{\rm int}$ in
(\ref{eq-hc2-tf:int}), $\FF$ the vector field introduced in
\eqref{eq-hc2-F}, and the function $\varphi_0$ is to be
specified later. Since $1\ll R\ll\varepsilon^{-\frac1{1+\rho}}$, we see that
$\psi_{\rho,\varepsilon}^{\rm bnd}$ and
$\psi_{\rho,R,\varepsilon}^{\rm int}$ have disjoint supports, hence
$$\mathcal E(\psi,\FF;\Omega)=\mathcal
E\left(\psi_{\rho,\varepsilon}^{\rm bnd},\FF;\Omega\right)+\mathcal
E\left([\mu]^{1/2}_+e^{-i\kappa
H\varphi_0}\psi_{\rho,R,\varepsilon}^{\rm
int},\FF;\Omega\right)\,.$$ We impose the condition
$\rho\in(0,\frac12)$. Then, thanks to Lemma~\ref{lem-hc2-ub:bnd}, we
get the following upper bound, 
$$\mathcal
E\left(\psi_{\rho,\varepsilon}^{\rm
bnd},\FF;\Omega\right)\leq-E_1|\partial\Omega|\kappa+o(\kappa).$$ So
we need to estimate the  term $E\left([\mu]^{1/2}_+e^{-i\kappa
H\varphi_0}\psi_{\rho,R,\varepsilon}^{\rm int},\FF;\Omega\right)$.

Recall the vector potential $\Ab_0$ introduced in
(\ref{eq-hc2-mpA0}). Notice that $\curl \FF=\curl \Ab_0=1$. So,
defining the function $\varphi_0$ by
\begin{equation}\label{eq-hc2-F=A}
-\Delta\varphi_0={\rm div}\,\Ab_0=0\quad{\rm in}~\Omega\,,\quad
\nu\cdot\nabla\varphi_0=\nu\cdot\Ab_0\quad{\rm
on}~\partial\Omega\,,\end{equation} we get
$\FF=\Ab_0-\nabla\varphi_0$. Therefore,
$$\mathcal E\left([\mu]^{1/2}_+e^{-i\varphi_0}\psi_{\rho,R,\varepsilon}^{\rm
int},\FF;\Omega\right)=\mathcal
E\left([\mu]^{1/2}_+\psi_{\rho,R,\varepsilon}^{\rm
int},\Ab_0;\Omega\right)\,.$$ Thanks to Lemma~\ref{lem-hc2-up:int}
and the definition of $E_2$ in \eqref{eq-hc2-E2}, we get,
$$\mathcal E\left([\mu]^{1/2}_+\psi_{\rho,R,\varepsilon}^{\rm
int},\Ab_0\right)\leq -[\mu]^2_+E_2|\Omega|\kappa+
C[\mu]_+\left(\kappa^{-1/2}R^\rho+\kappa^{3/2}R^{-\rho}\right)
+o([\mu]_+^2\kappa)\,.$$ Remembering the condition  $\rho\in
(0,\frac12)$ and taking $R=2\pi\sqrt{[ \kappa^\rho]+1}$, we get,
$$\mathcal E\left([\mu]^{1/2}_+\psi_{\rho,R,\varepsilon}^{\rm
int},\Ab_0\right)\leq
-E_2|\Omega|\,[\mu]^2_+\kappa+o(\max(1,[\mu]^2_+)\kappa)\,.$$ This
finishes the proof of Theorem~\ref{thm-hc2-FK-ub}.

\section{Lower bound of the energy}\label{hc2-sec-lb}

Let us pick a minimizer $(\psi,\Ab)$ of the Ginzburg-Landau energy
(\ref{eq-hc2-GL}). Our aim in this section is to give a lower bound
of the  energy $\mathcal E(\psi,\Ab;\Omega)$. We recall the convention that 
an open subset $D\subset\Omega$ is smooth if there exists an open and 
smooth set
$\widetilde D\subset\R^2$ such that $D=\widetilde D\cap\Omega$. For
all $a>0$, we
associate with a subset $D\subset\Omega$ the following subset of
$\Omega$,
\begin{equation}\label{eq-hc2-Da}
D_a=\{x\in\Omega~:~{\rm dist}(x,D)\leq a\}\,.\end{equation}
We will prove the following theorem.

\begin{thm}\label{thm-hc2-lb-loc}
Assume the magnetic field satisfies
$H=\kappa-\mu(\kappa)\sqrt{\kappa}$ with
$\displaystyle\lim_{\kappa\to\infty}
\frac{\mu(\kappa)}{\sqrt{\kappa}}=0$ and $\displaystyle\liminf_{\kappa\to\infty}\mu(\kappa)>-\infty$. Let $D\subset\Omega$ be  smooth,
open, and $\R_+\ni\kappa\mapsto a(\kappa)\in\R_+$ a function
satisfying $\displaystyle\lim_{\kappa\to\infty}a(\kappa)=0$. 

Then, for
any minimizer $(\psi,\Ab)$ of the energy $\mathcal E$ in
\eqref{eq-hc2-GL} and any continuous function $h\in C(\Omega)$  satisfying
$\|h\|_{L^\infty(\Omega)}\leq1$ and ${\rm supp}\,h\subset \overline D_a$, 
the following asymptotic lower bound
holds,
\begin{equation}\label{eq-hc2-lb-loc}
\mathcal E(h\psi,\Ab;\Omega)\geq-E_1|\overline
D\cap\partial\Omega|\kappa-E_2|D|\,[\mu(\kappa)]_+^2\kappa+o(\max(1,[\mu(\kappa)]_+^2)\kappa)\,,\quad{\rm
  as~}\kappa\to\infty\,.
\end{equation}
Here $E_1$ and $E_2$ are the constants 
introduced in (\ref{eq-hc2-E1}) and (\ref{eq-hc2-E2}) respectively.
\end{thm}

Recall the
ground state energy $C_0(\kappa,H)$ introduced in \eqref{eq-GL-GS}.
As corollary from Theorem~\ref{thm-hc2-lb-loc}, we get an asymptotic
lower bound for $C_0(\kappa,H)$.

\begin{corol}\label{thm-hc2-lb}
Assume $H=\kappa-\mu(\kappa)\sqrt{\kappa}$ and
$\frac{\mu(\kappa)}{\sqrt{\kappa}}\to0$ as $\kappa\to\infty$.
The following lower bound holds for the ground state energy in
\eqref{eq-GL-GS},
\begin{equation}\label{eq-hc2-lb-new}
C_0(\kappa,H)\geq
-E_1|\partial\Omega|\kappa
-E_2|\Omega|\,[\mu(\kappa)]_+^2\kappa+o(\max(1,[\mu(\kappa)]_+^2)\kappa)
\,,\quad{\rm as~}\kappa\to\infty\,.\end{equation}
Here $E_1>0$ and $E_2>0$ are the constants
introduced in (\ref{eq-hc2-E1}) and (\ref{eq-hc2-E2}) respectively.
\end{corol}
\begin{proof}
Assume the conclusion of the corollary were false and let
$\{\kappa_n\}$ be a sequence such that $\kappa_n\to\infty$ and
$\mu(\kappa_n)\to\mu_0$ as $n\to\infty$, with
$\mu_0\in\R\cup\{\pm\infty\}$.

 If $\mu_0=-\infty$, we apply Theorem~\ref{thm-hc2-Pan} (which is
 proved by Pan in \cite{Pa02}), with $\kappa=\kappa_n$ and $D=\Omega$.  Otherwise, if
 $\mu_0\in
\R\cup\{\infty\}$,
we apply Theorem~\ref{thm-hc2-lb-loc} with $D=\Omega$, 
$h\equiv1$ and $\kappa=\kappa_n$. In both cases we get a contradiction
to the assumption that the conclusion of the corollary were fasle.\end{proof}

The general strategy for proving Theorem~\ref{thm-hc2-lb-loc} 
is the following. Using a partition
of unity we may split the energy into a {\it boundary} component and
a {\it bulk} component. To control the  boundary component of the
energy we follow essentially the argument of Pan \cite{Pa02}.
However, the control of the bulk component of the energy is novel.
Finally, we make use of the {\it a priori} estimates recalled in
Lemma~\ref{lem-hc2-FoHe} and the improved estimate in
Corollary~\ref{cor-hc2:L2} in order to control the errors resulting
from the approximations (this is one additional key point that
replaces the implementation of the exponential decay of $\psi$ in
Pan's argument \cite{Pa02}).

In what follows, we consider a domain $D=\widetilde D\cap\Omega$, with
$\widetilde D$ a smooth and open domain in $\R^2$. We assume for
simplicity the following condition on $D$,
\begin{equation}\label{eq-hc2-cond-D}
\overline D 
{\rm ~is~ connected,}\quad 
\overline D\cap\partial\Omega=\partial\Omega,\quad |D|\not=|\Omega|\,.\end{equation} 
In general, $\overline D$ consists 
of a finite number of connected components.
By working with each component separately and adding up the
corresponding lower bounds, one can reduce to the connected case.
Furthermore,
$\overline D\cap\partial\Omega$ consists of a finite union of 
closed curves (possibly of zero length). 
A simple modification of the argument below will handle this case as
well, so we therefore only treat
domains $D$ satisfying the condition in \eqref{eq-hc2-cond-D}.

\subsection{Splitting of the energy}
Let us consider a parameter $\eta=\eta(\kappa)>0$ such that
$\eta\to0$ as $\kappa\to\infty$. We will take $\eta$  in the following form
\begin{equation}\label{eq-hc2-eta0}
\eta=\kappa^{-\rho}\,,\quad\rho\in(\frac14,1)\,,
\end{equation}
and we will fix a choice of $\rho$ at
the end of the proof.

Let us also consider  a partition of unity
$$\chi_1^2+\chi_2^2=1\quad{\rm in}~\R\,,\quad{\rm
supp}\,\chi_1\subset(-\infty,2]\,,\quad{\rm
supp}\,\chi_2\subset[1,\infty)\,.$$
We define,
\begin{equation}\label{eq-hc2-partition}
\chi_{1,\eta}(x)=\chi_1\left(\frac{\dist(x,\partial\Omega)}{\eta}\right)
\,,\quad
\chi_{2,\eta}(x)=\chi_2\left(\frac{\dist(x,\partial\Omega)}{\eta}\right)
\,,\quad\forall~x\in\Omega\,.
\end{equation}
Then we get the following localization formula,
\begin{equation}\label{eq-hc2-localization}
\int_\Omega|(\nabla-i\kappa H\Ab)h\psi|^2\,dx=
\sum_{j=1}^2\int_\Omega\left(|(\nabla-i\kappa
H\Ab)\chi_{j,\eta}h\psi|^2-\big{|}\,|\nabla\chi_{j,\eta}|h\psi\,\big{|}^2\right)\,dx\,.
\end{equation}
Defining the reduced energy,
\begin{equation}\label{eq-hc2-GL-A=0}
\mathcal E_0(\psi,\Ab;\Omega)=\int_\Omega\left(|(\nabla-i\kappa
H\Ab)\psi|^2-\kappa^2|\psi|^2+\frac{\kappa^2}{2}|\psi|^4\right)\,dx\,,\end{equation}
we then get in light of (\ref{eq-hc2-localization}),
\begin{equation}\label{eq-hc2-split:energy}
\mathcal E(h\psi,\Ab;\Omega)\geq \sum_{j=1}^2\mathcal
E_0(\chi_{j,\eta}h\psi,\Ab;\Omega)-\mathcal R(h\psi,\Ab)\,,
\end{equation}
where
\begin{equation}\label{eq-hc2-split:error}
\mathcal
R(h\psi,\Ab)=\sum_{j=1}^2\int_\Omega\big{|}\,|\nabla\chi_{j,\eta}|h\psi\,\big{|}^2\,
dx\,.\end{equation} This last term corresponds to an error which we
will estimate using Theorem~\ref{thm:Linfty-bulk}. Also, recall the
assumption that $|h|\leq1$ in $\Omega$. In this  way we get,
$$
\mathcal R(h\psi,\Ab)\leq
\|\nabla\chi_{j,\eta}\|^2_{L^\infty(\Omega)}\|\psi\|^2_{L^\infty(\omega_\kappa)}\int_{\{\eta\leq{\rm dist}(x,\partial\Omega)\leq2\eta\}}dx
\leq
C\,.
$$
Here $\omega_\kappa=\{x\in\Omega~:~{\rm dist}(x,\partial\Omega)\geq\kappa^{-\rho}\}$ and
$$\zeta=\zeta(\kappa)=\max\left\{\left|1-\frac{\kappa}H\right|^{1/2},
\kappa^{-1/4}\right\}\,.$$ Upon substitution into
(\ref{eq-hc2-split:energy}) we get,
\begin{equation}\label{eq-hc2-split:energy'}
\mathcal E(h\psi,\Ab;\Omega)\geq \sum_{j=1}^2\mathcal
E_0(\chi_{j,\eta}h\psi,\Ab;\Omega)-\frac{C}{\eta}\zeta^2\,.
\end{equation}
We proceed to estimate separately  the terms $\mathcal
E_0(\chi_{1,\eta}h\psi,\Ab;\Omega)$ and $\mathcal
E_0(\chi_{2,\eta}h\psi,\Ab;\Omega)$.

\subsection{Estimating the boundary energy}
Let us now introduce a further partition of unity\,,
$$h_1^2+h_2^2=1\quad{\rm
in~}\R\,,\quad {\rm supp}\,h_1\subset(-\infty,1)\,,\quad
{\rm supp}\,h_2\subset (-1,\infty)\,,$$
and such that
$$
h_1=1\quad{\rm in~}(-\infty,-1]\,,\quad h_2=1\quad{\rm in}~[1,\infty)\,.$$
Let $s_0=|\partial\Omega|/4$. Recall the  coordinate
transformation $\Phi$ in (\ref{eq:19}) valid in the neighborhood
$\Omega(t_0)$ of $\partial\Omega$. By defining,
$$h_{1,\eta}(x)=h_1\left(\frac{|s(x)|-s_0}{\eta}\right)\,,\quad
h_{2,\eta}(x)=h_2\left(\frac{|s(x)|-s_0}{\eta}\right)
\quad\forall~x\in\Omega(\eta)\,,$$ we get a partition of unity in
$\Omega(\eta)$. Using the localization  formula, the
energy splits one more time as follows,
\begin{equation}\label{eq-hc2-split:bnd00}
\mathcal E_0(\chi_{j,\eta}h\psi,\Ab;\Omega)= \sum_{j=1}^2\left(\mathcal
E_0(\psi_{j,\eta},\Ab;\Omega)
-\int_\Omega\big{|}|\nabla h_{j,\eta}|\chi_{1,\eta}h\psi\big{|}^2\,dx\right) \,,
\end{equation}
where
\begin{equation}\label{eq-hc2-partition'}
\psi_{j,\eta}(x)=h_{j,\eta}(x)\chi_{1,\eta}(x)h(x)\psi(x)\,,
\quad\forall~x\in\Omega(\eta)\,.
\end{equation}
Noticing that the supports of the
functions $|\nabla h_{j,\eta}|\chi_{1,\eta}\psi$, $j=1,2$, are
contained in
$\{x\in\Omega~:~{\rm dist}(x,\partial\Omega)\leq2\eta,~0\leq \big{|}
|s(x)|-s_0\big{|}\leq{\eta}\}$, we obtain for a possibly new constant $C$,
$$\sum_{j=1}^2\int_\Omega\big{|}|\nabla h_{j,\eta}|
\chi_{1,\eta}h\psi\big{|}^2\,dx\leq C\eta\,.$$
Substituting this into \eqref{eq-hc2-split:bnd00}, we obtain,
\begin{equation}\label{eq-hc2-split:bnd}
\mathcal E_0(\chi_{j,\eta}h\psi,\Ab;\Omega)\geq \sum_{j=1}^2\mathcal
E_0(\psi_{j,\eta},\Ab;\Omega)-C\,.
\end{equation}

Let us now bound the term $\mathcal
E_0(\psi_{1,\eta},\Ab;\Omega)$ from below. Since $\curl A=1$ on
$\partial\Omega$, Lemma~F.1.1 of \cite{FoHe08} yields that up to a
gauge transformation, we may write in $(s,t)$-coordinates,
$$\widetilde{\Ab}
(s,t)=\left(-t+\frac{t^2}2k(s)+t^2b(s,t)\,,\,0\right)\quad{\rm
in}~{\rm supp}\,h_{1,\eta}\,,$$ where the function $b$ satisfies
$$|b(s,t)|\leq \|\curl\Ab-1\|_{C^1(\Omega)}\leq
C\kappa^{-1}\,,\quad\forall~(s,t)\in{\rm supp}\,h_{1,\eta}\,,
$$
and where we used Lemma~\ref{lem-hc2-FoHe} in the last step. We also
remind the reader that $\widetilde \Ab=U_\Phi\Ab$ is the vector
field associated to $\Ab$ by the coordinate transformation $\Phi$.

It is more convenient in this part to introduce the parameters,
$$\varepsilon=\frac1{\sqrt{\kappa
H}}\,,\quad\gamma=\frac{\kappa}H-1
=\mu(\kappa)\frac{\sqrt{\kappa}}{H}=\mathcal
O\left(\mu\sqrt{\varepsilon}\right)\,.$$ Then, applying the change
of variables formula in Proposition~\ref{hc2-App:transf}, we get,
\begin{equation}\label{eq-hc2-lb:bnd}
\mathcal E_0(\psi_{1,\eta},\Ab;\Omega)=\mathcal
J_\varepsilon\left(\widetilde\psi_{1,\eta}\right)\,,
\end{equation}
where
\begin{multline}\label{eq-hc2-lb:Q}
\mathcal
J_\varepsilon\left(\widetilde\psi_{1,\eta}\right)=\int_{-\frac{|\partial\Omega|}2}^{\frac{|\partial\Omega|}2}
\int_0^{t_0} \left[
(1-tk(s))^{-2}\left|\left(\partial_s+\frac{i}{\varepsilon^2}(t+\widetilde
b(s,t))\right)\widetilde \psi_{1,\eta}\right|^2\right.\\
\left.+\left|\partial_t\widetilde\psi_{1,\eta}\right|^2-\frac{1+\gamma}{\varepsilon^2}
|\widetilde\psi_{1,\eta}|^2+\frac{1+\gamma}{2\varepsilon^2}
|\widetilde\psi_{1,\eta}|^4\right](1-tk(s))\,dsdt\,,
\end{multline}
and
\begin{equation}\label{eq-hc2-b(s,t)}
\widetilde b(s,t)=-t^2\left(\frac{k(s)}2+b(s,t)\right)=\mathcal O
(t^2)\,. \end{equation} Let us introduce another parameter
$\delta=\delta(\varepsilon)\in(0,1)$ to be fixed later. Applying a
Cauchy-Schwarz inequality, we get,
$$\left|\left(\partial_s+\frac{i}{\varepsilon^2}(t+\widetilde
b(s,t))\right)\widetilde \psi_{1,\eta}\right|^2\geq
(1-\delta)\left|\left(\partial_s+\frac{i}{\varepsilon^2}t\right)\widetilde
\psi_{1,\eta}\right|^2-\delta^{-1}\frac{1}{\varepsilon^4}\left|\,|\widetilde
b(s,t)|^2\,\widetilde \psi_{1,\eta}\right|^2\,.$$
Substituting the above estimate into \eqref{eq-hc2-lb:Q}, we get,
\begin{equation}\label{eq-hc2-bnd:mainterm}
\mathcal J_\varepsilon\left(\widetilde\psi_{1,\eta}\right)\geq
(1-\delta)\mathcal
Q_\varepsilon\left(\widetilde\psi_{1,\eta}\right) -C\,\mathcal
R_{\rm bnd}\left(\widetilde\psi_{1,\eta}\right)\,,
\end{equation}
where
$$\mathcal Q_\varepsilon\left(\widetilde\psi_{1,\eta}\right)
= \int_{-\frac{|\partial\Omega|}2}^{\frac{|\partial\Omega|}2}
\int_0^{t_0} \left(
\left|\left(\partial_s+\frac{i}{\varepsilon^2}t\right)\widetilde
\psi_{1,\eta}\right|^2+\left|\partial_t\widetilde\psi_{1,\eta}\right|^2-\frac{1}{\varepsilon^2}
|\widetilde\psi_{1,\eta}|^2+\frac1{2\varepsilon^2}|\widetilde\psi_{1,\eta}|^4\right)\,dsdt\,,
$$
and \begin{multline*} \mathcal R_{\rm
bnd}\left(\widetilde\psi_{1,\eta}\right)=\frac{1}{\varepsilon^2}\int\left[
|\gamma|+t+\frac{\delta^{-1}}{\varepsilon^2}t^4\right]|\widetilde\psi_{1,\eta}|^2\,dsdt+\\
\int
t\left(\left|\left(\partial_s+\frac{i}{\varepsilon^2}(t+\widetilde
b(s,t))\right)\widetilde
\psi_{1,\eta}\right|^2+\left|\partial_t\widetilde
\psi_{1,\eta}\right|^2\right)\,dsdt\,.
\end{multline*}
From the definition of $\psi_{1,\eta}$, we know that
$${\rm
supp}\,\psi_{1,\eta}\subset\{x\in\Omega~:~\dist(x,\partial\Omega)\leq\eta\,,
~
|s(x)|\leq\frac{|\partial\Omega|}4+\eta\}\,,
\quad|\psi_{1,\eta}(x)|\leq|\psi(x)|\leq1\,.$$
With this point on the one  hand, and (\ref{eq-hc2-FoHe2}) on the
other hand, we deduce that,
\begin{equation}\label{eq-hc2-lb:bnd-error}
\mathcal R_{\rm bnd}\left(\widetilde\psi_{1,\eta}\right)\leq
\frac{C}{\varepsilon^2}\left(\eta|\gamma|+\eta^2+\frac{\delta^{-1}}{\varepsilon^2}\eta^5\right)\,.
\end{equation}
Let us define the re-scaled function,
$$g_\eta(\sigma,\tau)=\left\{
\begin{array}{ll}
\widetilde\psi_{1,\eta}\left(\displaystyle\varepsilon{\sigma},
\displaystyle\varepsilon{\tau}\right)&{\rm
if}~ (\sigma,\tau)\in (-\ell,\ell)\times(0,\eta\varepsilon^{-1})\,,\\
0&{\rm otherwise}\,,\end{array}\right.$$ where
$$\ell=\frac{|\partial\Omega|}{4\varepsilon}+\frac{\eta}{\varepsilon}\,.$$
In the new scale, we may write,
$$\mathcal
Q_\varepsilon(\widetilde\psi_{1,\eta})=\mathcal
E_\ell(g_\eta)\,,$$ where $\mathcal E_\ell$ is the functional
introduced in \eqref{eq-hc2-redGL}. Invoking
Theorem~\ref{thm-hc2-Pa02}, we get a new constant $M>0$ such that,
\begin{equation}\label{eq-hc2-lb:bnd-refGL}
\mathcal Q_\varepsilon(\widetilde\psi_{1,\eta})\geq d(\ell)\geq
-E_1\left(\frac{|\partial\Omega|}{2\varepsilon}+2\frac\eta\varepsilon\right)-M\,.
\end{equation}
Summing up the estimates in \eqref{eq-hc2-lb:bnd},
\eqref{eq-hc2-bnd:mainterm}, \eqref{eq-hc2-lb:bnd-error} and
\eqref{eq-hc2-lb:bnd-refGL}, we get finally,
$$\mathcal E_0(\psi_{1,\eta},\Ab;\Omega)\geq
-\left(E_1\frac{|\partial\Omega|}{2\varepsilon}+C\frac\eta\varepsilon
+M\right)(1-\delta)-
\frac{C}{\varepsilon^2}\left(\eta^2+\eta|\gamma|+\frac{\delta^{-1}}{\varepsilon^2}\eta^5\right)
\,.
$$
In a similar fashion, we establish that,
$$
\mathcal E_0(\psi_{2,\eta},\Ab;\Omega)\geq
-\left(E_1\frac{|\partial\Omega|}{2\varepsilon}+C\frac\eta\varepsilon
+M\right)(1-\delta)-
\frac{C}{\varepsilon^2}\left(\eta^2+\eta\gamma+\frac{\delta^{-1}}{\varepsilon^2}\eta^5\right)
\,.$$
Invoking \eqref{eq-hc2-split:bnd}, and  recalling the
definition of $\gamma$ and $\varepsilon=\kappa^{-1}(1+o(1))$, we
get,
\begin{equation}\label{eq-hc2-lb:bnd*}
\mathcal E_0(\chi_{1,\eta}h\psi,\Ab;\Omega)\geq
-E_1|\partial\Omega|\kappa- C\kappa^2 \left(\eta^2
+\frac{\eta|\mu(\kappa)|}{\sqrt{\kappa}}+\delta^{-1}\kappa^2\eta^5\right)
-C(\delta+\eta)\kappa\,.
\end{equation}

\subsection{Estimating the bulk energy}
We recall that, for a given $R>0$, we denote by $K_{R}$ the unit
square of the lattice $R(\mathbb Z\oplus i\mathbb Z)$.

For $x\in\R^2$ and $R>0$, we denote by $K_{R}(x)$ a square of {\bf
center} $x$ and side length $R$,
\begin{equation}\label{eq-hc2-unitsquare}
K_{R}(x)=\left(x_1-\frac{R}2\,,\,x_1+\frac{R}2\right)\times\left(x_2-\frac{R}2\,,\,x_2+\frac{R}2\right)\,,
\quad\forall~x=(x_1,x_2)\in\R^2,~R>0\,.
\end{equation}
Let us consider a fixed number $\alpha\in(0,\frac12)$, and cover
$\R^2$ by the squares $(K_1(x_{j,\alpha}))_{j\in\mathbb
Z\oplus\mathbb Z}$ where
\begin{equation}\label{eq-hc2-xa}
x_{j,\alpha}=(1-\alpha)j\,,\quad\forall~j=(j_1,j_2)\in\mathbb
Z\times\mathbb Z\,.
\end{equation}
Let us take a partition of unity in $\R^2$ associated with the
squares $K_1(x_{j,\alpha})$,
$$\sum_{j}u_j^2=1\,,\quad {\rm supp}\,u_j\subset
K_1(x_{j,\alpha})\,.
$$
Recall the parameter $\varepsilon=1/\sqrt{\kappa H}$. Let us
consider a further parameter $R=R(\varepsilon)>1$  such that
$$R(\varepsilon)\to\infty,\quad \varepsilon
R(\varepsilon)\to0\quad{\rm as~}\varepsilon\to0,\quad{\rm and}~ R^2\in2\pi\mathbb N\,.$$
Then,  defining,
$$u_{j,R}(x)=u_j\left(\frac{
    x}{\varepsilon R}\right)\,,\quad\forall~x\in\Omega\,,$$
we get a partition of unity associated with the re-scaled squares $K_{\varepsilon R}(x^\varepsilon_{j,\alpha})$
\begin{equation}\label{eq-hc2-cover:e}
x^\varepsilon_{j,\alpha}=\varepsilon R
x_{j,\alpha}=(1-\alpha)\varepsilon R j\,,\quad j\in\mathbb
Z\times\mathbb Z\,,
\end{equation}
 such that,
$${\rm supp}\,u_{j,R}\subset K_{\varepsilon
R}(x^\varepsilon_{j,\alpha}).$$ In order to simplify notation, we
will skip the dependence on $\varepsilon$ and $\alpha$ from the
squares and write
$K_{\varepsilon R}^j$
instead of $K_{\varepsilon R}(x^\varepsilon_{j,\alpha})$,
$j\in\mathbb Z^2$.

Let us introduce
 further,
\begin{equation}\label{eq-hc2-index}
 \mathcal J=\mathcal
 J_{\varepsilon,D_a}=\{j~:~D_a\cap{\rm supp}\,u_{j,R}\not=\emptyset\}\,,\quad
 N_\varepsilon={\rm Card}\,\mathcal J\,.\end{equation}
Here $a=a(\kappa)\to0$ as $\kappa\to\infty$ and $D_a$ is the
neighborhood of $D$ introduced in \eqref{eq-hc2-Da}. 

We notice that,
\begin{equation}\label{eq-hc2-Nindex}
\lim_{\varepsilon\to0} \left(N_\varepsilon\times
  \varepsilon^2 R^2\right)=(1+v(\alpha))|D|\,,
\end{equation}
where $v(\alpha)$ is positive and verifies 
(actually $v(\alpha)=\alpha(1-\alpha)$),
\begin{equation}\label{eq-hc2-v(alpha)}
\lim_{\alpha\to0}v(\alpha)=0\,.
\end{equation}
 Implementing the partition of unity $u_{j,R}$, we get a splitting of the interior energy,
\begin{equation}\label{eq-hc2-lb:int}
\mathcal E_0(\chi_{2,\eta}h\psi,\Ab;\Omega)\geq \sum_{j\in\mathcal J}\mathcal
E_0(u_{j,R}\chi_{2,\eta}h\psi,\Ab;\Omega) -\frac{C}{(\varepsilon
R)^2}\|\psi\|^2_{L^\infty(\{x\in\Omega~:~{\rm dist}(x,\partial\Omega)\geq\kappa^{-\rho}\})}\,.
\end{equation}
Setting
$$\varphi_{j,R}(x)=
u_{j,R}(x)\chi_{2,\eta}(x)h(x)\psi(x)\,,\quad\forall~x\in\Omega\,,
$$
and invoking again Theorem~\ref{thm:Linfty-bulk}, we infer from
\eqref{eq-hc2-lb:int},
\begin{equation}\label{eq-hc2-lb:int'}
\mathcal E_0(\chi_{2,\eta}h\psi,\Ab;\Omega)\geq \sum_{j\in\mathcal J}\mathcal
E_0(\varphi_{j,R},\Ab;\Omega) -\frac{C}{(\varepsilon
R)^2}\zeta^2\,,
\end{equation}
and we point out that the constant $C$ depends on the parameter
$\alpha$, but we will not need to make this dependence explicit in
the notation as $\alpha$  remains fixed in the limit
$\varepsilon\to0$. We also remind the reader that
$\zeta=\max\{|1-\frac{\kappa}H|^{1/2},\kappa^{-1/4}\}$.

Let us proceed to estimate $\mathcal
E_0(\varphi_{j,R},\Ab;\Omega)$. We apply first a gauge transformation
that allows us to approximate the
vector field $\Ab$ locally.
 Setting
$$\mathbf B(x)=\curl\,\Ab(x)\,,$$
then there exists a real-valued function $\phi_0$
such that we may write,
$$\Ab(x)-\nabla\phi_0
=\int_0^1 sB(sx)(-x_2,x_1)\,ds\,,\quad\forall~x=(x_1,x_2)\in\Omega\,.$$
Notice that $\varphi_{j,R}$ is
supported in a ball $B(x_j,C\varepsilon R)$ with $C$ sufficiently
large.
We may write,
$$\Ab_0(x)-\nabla\phi_0=\frac{\mathbf
  B(x_j)}2(-x_2,x_1)+a(x)\,,\quad{\rm in}~B(x_j,C\varepsilon R)\,,$$
where the vector field $a(x)$ satisfies the uniform estimate,
$$|a(x)|\leq C\|\nabla\mathbf
B\|_{L^\infty(\Omega)}|x-x_j|\,,\quad{\rm in~}B(x_j,C\varepsilon R)\,.$$
Therefore, we can find a real-valued function $\phi_{j,R,\varepsilon}$
such  that,
$$\Ab_0(x)-\nabla\phi_{j,R,\varepsilon}=\frac{\mathbf
  B(x_j)}2(x-x_j)^\bot+a(x)\,,\quad{\rm in}~B(x_j,C\varepsilon R)\,,$$
where $x^\bot=(-x_2,x_1)$ for all $x=(x_1,x_2)\in\R^2$.

By Lemma~\ref{lem-hc2-FoHe}, the magnetic field $\mathbf B$ is
almost equal to $1$, hence we get,
\begin{equation}\label{eq-hc2-lb:A=A0}
\left|\Ab(x)-\nabla\phi_{j,R,\varepsilon}(x)-\Ab_0(x)\right|
\leq C\kappa^{-1}|x-x_j|\,,\quad{\rm in~}B(x_j,C\varepsilon R)\,.
\end{equation}
Here $\Ab_0$ is the  vector field  introduced in
\eqref{eq-hc2-mpA0}. Therefore, setting
$$
\Ab_{j,R}(x)=\Ab(x)-\nabla\phi_{j,R,\varepsilon}( x)\,,$$ we infer
from \eqref{eq-hc2-lb:A=A0},
\begin{equation}\label{eq-hc2-lb:A=A0'}
\left|\Ab_{j,R}(x)-\Ab_{0}(x)\right|\leq
C\varepsilon|x-x_j|\,,\quad{\rm in~}{\rm
    supp}\,
\varphi_{j,R}\,.
\end{equation}
Furthermore, we notice that,
\begin{equation}\label{eq-hc2-lb:gt}
\mathcal E_0(\varphi_{j,R},\Ab;\Omega)= \mathcal
E_0(e^{i\phi_{j,R,\varepsilon}}\varphi_{j,R},\Ab_{j,R};\Omega)\,.
\end{equation}
Using a Cauchy-Schwarz inequality, we get for any
$\beta\in(0,1)$,
\begin{multline*}
\mathcal
E_0(e^{i\phi_{j,R,\varepsilon}}\varphi_{j,R},\Ab_{j,R};\Omega)\geq
(1-\beta)\mathcal
E_0(e^{i\phi_{j,R,\varepsilon}}\varphi_{j,R},\Ab_0;\Omega)
\\-C\beta^{-1}\varepsilon^{-4}\int_\Omega|\Ab_{j,R}(x)-\Ab_0(x)|^2\,|\varphi_{j,R}|^2\,dx\,.\end{multline*}
We implement \eqref{eq-hc2-lb:A=A0'} in the above estimate and
we use the bound $|\varphi_{j,\eta}(x)|\leq \zeta$. That way we get,
\begin{equation}\label{eq-hc2-lb:gt'}
\mathcal
E_0(e^{i\phi_{j,R,\varepsilon}}\varphi_{j,R},\Ab_{j,R};\Omega)
\geq(1-\beta)\mathcal
E_0(e^{i\phi_{j,R,\varepsilon}}\varphi_{j,R},\Ab_{0};\Omega)
-C\beta^{-1}\varepsilon^2R^4\zeta^2\,.
\end{equation}
We proceed to obtain a lower bound for $\mathcal E_0
 (e^{i\phi_{j,R,\varepsilon}}\varphi_{j,R},\Ab_{0};\Omega)$. Modulo a translation, we
 may assume that ${\rm supp}\,\varphi_{j,R}$  is contained in the unit
 square $K_{\varepsilon R}$ of the lattice
 $\varepsilon R(\mathbb Z\oplus i\mathbb Z)$. We therefore define the
 re-scaled function,
\begin{equation}\label{eq-hc2-lb:newf}
f_j(\widetilde x)=\left(e^{i\phi_{j,R,\varepsilon}}\varphi_{j,R}\right)
\left(\varepsilon \widetilde x\right)\,,\quad\forall~\widetilde x\in K_{R}\,.
\end{equation}
That way the energy becomes (after omitting the tildes from the notation),
$$
\mathcal E_0
 (e^{i \phi_{j,R,\varepsilon}}\varphi_{j,R},\Ab_{0};\Omega)
=\int_{K_{R}} \left(|(\nabla-i\Ab_0)f_j|^2-\frac\kappa{H}|f_j|^2+
\frac{\kappa}{2H}|f_j|^4\right)\,dx\,.
$$
Invoking \eqref{eq-hc2-lb:gt} and \eqref{eq-hc2-lb:gt'}, we deduce
that,
\begin{multline}\label{eq-hc2-lb:newEn}
\mathcal E_0
 (\varphi_{j,R},\Ab;\Omega)
=(1-\beta)\int_{K_{R}}
\left(|(\nabla-i\Ab_0)f_j|^2-\frac\kappa{H}|f_j|^2+
\frac{\kappa}{2H}|f_j|^4\right)\,dx\\
-C\beta^{-1}\varepsilon^2R^4\zeta^2\,.
\end{multline}

Up to now, we are able to prove the following lemma, whose proof is
actually a simple application of the Cauchy-Schwarz inequality.

\begin{lem}\label{thm-hc2-lb:1}
Assume that $H=\kappa-\mu(\kappa)\sqrt{\kappa}$ with
$\displaystyle\limsup_{\kappa\to\infty}\mu(\kappa)=\mu_0$ and 
$\mu_0\in[-\infty,0]$.
Then, as $\kappa\to\infty$,
$$\mathcal E_0(\chi_{2,\eta}h\psi,\Ab;\Omega)\geq
-C\left(\beta^{-1}R^2+\frac1{(\varepsilon
R)^{2}}\right)\zeta^2+o(\kappa)\,.$$
\end{lem}
\begin{proof}
Notice that, for any $j$, $f_j$ can be considered as a function
 in the domain of the
periodic operator $P_{R}$, see \eqref{eq-hc2-poperator}. Using
Proposition~\ref{prop-hc2-poperator} and the variational min-max
principle, we write, \begin{multline}\label{eq-hc2-mu=0}\int_{K_{R}}
\left(|(\nabla-i\Ab_0)f_j|^2-\frac\kappa{H}|f_j|^2+
\frac{\kappa}{2H}|f_j|^4\right)\,dx\\
\geq\left(1-\frac{\kappa}{H}\right)
\int_{K_R}|f_j|^2\,dx+\frac{\kappa}{2H}\int_{K_R}|f_j|^4\,dx\,.\end{multline}
If $\mu(\kappa)\leq0$, i.e. $H\geq \kappa$, we have nothing to prove
since  the right hand side of \eqref{eq-hc2-mu=0} is positive, and
we only need to collect the estimates \eqref{eq-hc2-lb:int'} and
\eqref{eq-hc2-lb:gt}-\eqref{eq-hc2-lb:newEn}.

Now we assume that $\mu(\kappa)\to0$ as $\kappa\to\infty$. Using the
following Cauchy-Schwarz inequality
$$|\mu(\kappa)|\sqrt{\kappa}\int_{K_R}|f_j|^2\,dx\leq\frac14\kappa\int_{K_R}|f_j|^4+4R^2\mu^2(\kappa)\,,$$
we obtain from \eqref{eq-hc2-mu=0},
\begin{align*}
\int_{K_R}\left(|(\nabla-i\Ab_0)f_j|^2-\frac\kappa{H}|f_j|^2+
\frac{\kappa}{2H}|f_j|^4\right)\,dx&\geq
\frac{\kappa}{4H}\int_{K_R}|f_j|^4\,dx-\frac{4R^2\mu^2(\kappa)}{H}\\
&\geq -\frac{5R^2\mu^2(\kappa)}{H}\,.
\end{align*}
Summing over $j$ (recall that the number of indices $j$ is
proportional to $\varepsilon^{-2}R^{-2}\sim \kappa^2R^{-2}$), we get
$$\sum_j\left(|(\nabla-i\Ab_0)f_j|^2-\frac\kappa{H}|f_j|^2+
\frac{\kappa}{2H}|f_j|^4\right)\,dx\geq -C\mu^2(\kappa)\kappa\,.$$ To finish
the proof of the lemma, it suffices to collect the estimates
\eqref{eq-hc2-lb:int'} and
\eqref{eq-hc2-lb:gt}-\eqref{eq-hc2-lb:newEn}.
\end{proof}

We assume from now on that $ H=\kappa-\mu(\kappa)\sqrt{\kappa}$ with
$$\mu(\kappa)>0\,.$$
Let us introduce,
$$\mathcal J_+=\left\{j\in\mathcal J~:~\int_{K^j_{R}}
\left(|(\nabla-i\Ab_0)f_j|^2-\frac\kappa{H}|f_j|^2\right)\,dx\geq0\right\}\,,$$
and set
\begin{equation}\label{eq-hc2-lb:nep}
n_\varepsilon ={\rm Card}\,\mathcal J_+\,.
\end{equation}
We shall obtain a lower bound of $\mathcal
E_0(\chi_{2,\eta}h\psi,\Ab;\Omega)$ in terms of the `local energies'
associated with the indices $j$ that are not in $\mathcal J_+$\,.

Let us pick an arbitrary $j\not\in\mathcal J_+$. Then,
$$\int_{K_{R}}
\left(|(\nabla-i\Ab_0)f_j|^2-\frac\kappa{H}|f_j|^2\right)\,dx<0\,.$$
Notice that the function $f_j$ belongs to the domain of the periodic
operator $P_{R}$, whose first eigenvalue equals to $1$. Let us
recall that we introduced the following parameter
$$\gamma=\frac{\kappa}H-1\,.$$
With this in hand we may write,
$$Q_{R}(f_j)-(1+\gamma)\int_{K_{R}}|f_j|^2\,dx< 0\,.$$
Invoking Lemma~\ref{prop-hc2-poperator'}, we get,
\begin{equation}\label{eq-hc2-lb:f=proj'}
\|f_j-\Pi_1f_j\|_{L^4(K_{R})}\leq
C\sqrt{\gamma}\|f_j\|_{L^2(K_R)}.
\end{equation}
Here, recall the space $L_{R}$ introduced in \eqref{eq-hc2-space2}
and $\Pi_1$ the orthogonal projector on $L_{R}$.

It results from the triangle inequality and \eqref{eq-hc2-lb:f=proj'}
that
$$\|\Pi_1f_j\|_{L^4(K_R)}\leq
\|f_j\|_{L^4(K_R)}+C\sqrt{\gamma}\|f_j\|_{L^2(K_R)} \,.$$
Applying Cauchy-Schwarz inequality twice, we get for any $\sigma\in(0,1)$ the
following estimate,
\begin{equation}\label{eq-hc2-lb:f=proj''}
\|\Pi_1f_j\|^4_{L^4(K_R)}\leq
(1+\sigma)\|f_j\|^4_{L^4(K_R)}+C\sigma^{-3}
\gamma^2\|f_j\|^4_{L^2(K_R)}\,.
\end{equation}
Using the definition of $f_j$ and Theorem~\ref{thm:Linfty-bulk}, we
get that $|f_j|\leq\zeta$. Hence, we infer from
\eqref{eq-hc2-lb:f=proj''},
\begin{equation}\label{eq-hc2-lb:f=proj}
\|\Pi_1f_j\|^4_{L^4(K_R)}\leq
(1+\sigma)\|f_j\|^4_{L^4(K_R)}+C\sigma^{-3}\gamma^2
R^4\zeta^4\,.
\end{equation}
Using the mini-max principle and \eqref{eq-hc2-lb:f=proj}, we get,
\begin{align*}
&\hskip -1cm\int_{K_{R}}
\left(|(\nabla-i\Ab_0)f_j|^2-\frac\kappa{H}|f_j|^2+
\frac{\kappa}{2H}|f_j|^4\right)\,dx\\
&\geq\int_{K_{R}}\left[\left(1-\frac{\kappa}{H} \right)|\Pi_1f_j|^2
+\frac{\kappa}{2H}(1-\sigma) |\Pi_1f_j|^4\right]\, dx-C\sigma^{-3}
\gamma^2R^4\zeta^2\,.
\end{align*}
We choose $\sigma$ as function of $\varepsilon$ and we impose on it
the following condition,
\begin{equation}\label{eq-hc2-const-sigma}
\sigma\to0\quad{\rm as~}\varepsilon\to0\,.
\end{equation}
Therefore, defining
$$c=\sqrt
{\frac{\frac{\kappa}{H}(1-\sigma)}{\frac{\kappa}{H}-1}}
\,,\quad g(x)=c\, \Pi_1f_j\,,$$ we get,
\begin{align}\label{eq-hc2-lb:conc}
\int_{K_{R}} \left(|(\nabla-i\Ab_0)f_j|^2-\frac\kappa{H}|f_j|^2+
\frac{\kappa}{2H}|f_j|^4\right)\,dx&=
\frac{|K_{R}|}{c^2}\left(\frac{\kappa}{H}-1\right)F_{R}(g)\nonumber\\
&\geq |K_{R}|\frac{|\mu(\kappa)|^2}{\kappa}\big(1+o(1)\big)
c(R)\quad{\rm as~} \kappa\to\infty\,.
\end{align}
Here recall the energy $F_{R}$ and the constant $c(R)$ introduced in
\eqref{eq-hc2-eneAb} and \eqref{eq-hc2-c(r,t)} respectively.

Consequently, collecting \eqref{eq-hc2-lb:int'},
\eqref{eq-hc2-lb:gt}, \eqref{eq-hc2-lb:gt'},
\eqref{eq-hc2-lb:newEn}, \eqref{eq-hc2-lb:conc}, we get,
\begin{align}\label{eq-hc2-concl'0}
\mathcal
E_0(\chi_{2,\eta}h\psi,\Ab;\Omega)
\geq& (1-\beta)\bigg{[}(N_\varepsilon-n_\varepsilon)
\frac{|K_{R}|}{\kappa}\,|\mu(\kappa)|^2c(R)(1+o(1))
+\sum_{j\in \mathcal
J_+}\frac{\kappa}{2H}\int|f_j|^4\,dx\bigg{]}\nonumber\\
&%
-C\left[\left(\beta^{-1}R^2+\frac1{(\varepsilon R)^2}\right)
\zeta^2+\sigma^{-3}\varepsilon^{-2}\gamma^2R^2\zeta^4\right]\,.
\end{align}
Notice that, as a result of \eqref{eq-hc2-Nindex}, we have,
$$\frac{N_\varepsilon|K_{R}|}{\kappa}=\kappa\,c_\alpha|D|+o(\kappa)~{\rm
  as~} \kappa\to\infty\,,\quad
{\rm and }\quad c(R)=-E_2+o(1){~\rm as~}R\to\infty\,,$$ where
$c_\alpha=1+\alpha(1-\alpha)$
and the constant $\alpha\in(0,1)$ is introduced in
connection with the partition of unity $u_{j,R}$. Since $E_2>0$, we
deduce the following lower bound from \eqref{eq-hc2-concl'0},
\begin{align}\label{eq-hc2-concl'}
\mathcal
E_0(\chi_{2,\eta}h\psi,\Ab;\Omega)
\geq&
(1-\beta)\bigg{(}-c_\alpha|D|
E_2[\mu(\kappa)]_+^2\kappa+o(|\mu(\kappa)|^2\kappa)\bigg{)}
\nonumber\\
&
-C\left[\left(\beta^{-1}R^2+\frac1{(\varepsilon R)^2}\right)
\zeta^2+\sigma^{-3}\varepsilon^{-2}\gamma^2R^2\zeta^4\right].
\end{align}

\subsection{Proof of Theorem~\ref{thm-hc2-lb-loc}}
We recall that we only treat the case when the domain $D$ satisfies
the condition in \eqref{eq-hc2-cond-D}. We also recall the condition 
$\displaystyle\liminf_{\kappa\to\infty}\mu(\kappa)>-\infty$.

Assume that the conclusion of Theorem~\ref{thm-hc2-lb-loc} were
false. Then there exist sequences $\{\kappa_n\}$, $\{H_n\}$, 
$\mu_0\in\R\cup\{+\infty\}$, $c>0$ and a sequence of minimizers
$\{(\psi_n,\Ab_n)\}$ of the energy $\mathcal E$ such that,
$$\kappa_n\to\infty\quad\quad\frac{\kappa_n}{H_n}\to1\quad
\mu(\kappa_n)\to\mu_0\quad{\rm as ~}n\to\infty\,,$$
and
\begin{equation}\label{eq-hc2-assumption-n}
\mathcal E(h\psi_n,\Ab_n;\Omega)\leq -E_1|\Omega|\kappa_n
-E_2|D|\,[\mu(\kappa_n)]_+^2\kappa_n
-c\max(1,[\mu(\kappa_n)]_+^2)\kappa_n\,. \end{equation}
We treat seperately the two cases $\mu_0\in(-\infty,0]$ and
$\mu_0\in(0,\infty]$.
Assume that $\mu_0\in(-\infty,0]$.
Then
$\zeta=\max\{|1-\frac{\kappa_n}{H_n}|^{1/2},\kappa_n^{-1/4}\}=\kappa_n^{-1/4}$.
Combining \eqref{eq-hc2-split:energy'}, \eqref{eq-hc2-lb:bnd*}
and Lemma~\ref{thm-hc2-lb:1}, we get,
\begin{align*}
\mathcal E_0(h\psi_n,\Ab_n;\Omega)\geq&
-E_1|\partial\Omega|\kappa_n
-C\left[\kappa^2_n\left(\eta^2+\eta\frac{|\mu(\kappa_n)|}{\sqrt{\kappa_n}}+\delta^{-1}\kappa^2_n\eta^5\right)+(\delta+\eta)\kappa_n
+\frac{1}{\eta}\kappa^{-1/2}_n\right]\\
&-C\left(\beta^{-1}R^2+\frac1{(\varepsilon_n
R)^2}\right)\varepsilon^{1/2}_n +o(\kappa_n)\,.
\end{align*}
 Here $\varepsilon_n=\frac1{\sqrt{\kappa_nH_n}}$. We choose
$$\eta=\kappa^{-\rho}_n,\quad \delta=\sqrt{\kappa^3_n\eta^5}\,,\quad
\frac35<\rho<1\,,$$
$$\beta=\kappa^{-1/4}_n\,,$$
and $R=2\pi[\kappa^{3/4}_n]$, with $[x]$ denotes the integer
part of $x$.  In this
way 
we get,
$$\mathcal E_0(h\psi_n,\Ab_n;\Omega)\geq
-E_1|\partial\Omega|\kappa_n+o(\kappa_n)$$
thereby contradicting \eqref{eq-hc2-assumption-n}.

We  now treat  the  case $\mu_0\in(0,\infty]$.
In this case
$\zeta\approx \sqrt{\mu(\kappa_n)}\kappa^{-1/4}_n$ and
$\gamma=\frac{\kappa_n}{H_n}-1\approx\zeta^2$.

We  make  the following choice of the parameters:
$$\eta=\mu^{2/5}(\kappa_n)\kappa^{-\rho}_n,\quad \delta=\frac{\sqrt{\kappa^3_n\eta^5}}{\mu(\kappa_n)}\,,\quad
\frac35<\rho<1\,,$$
$$\beta=(\mu^2_n\varepsilon_n)^{3/8}\,,\quad\sigma=(\mu^2_n\varepsilon_n)^{1/16}\,,$$
and  $R=2\pi[(\mu^2_n\varepsilon_n)^{-3/8}]$. Here
$\mu_n=\mu(\kappa_n)$.
With this choice of parameters, we have,
$$
\frac1{(\varepsilon_n
  R)^2}\zeta^2+\beta^{-1}R^2\zeta^2+\sigma^{-3}\varepsilon^{-2}_n
\gamma^2R^2\zeta^4=o(\mu^2\varepsilon^{-1}_n)\quad{\rm as ~}
\varepsilon_n\to0\,.$$
Therefore, we get by
 combining \eqref{eq-hc2-split:energy'}, \eqref{eq-hc2-lb:bnd*}
and  \eqref{eq-hc2-concl'},
\begin{equation}\label{eq-hc2-conc''}
\mathcal
E_0(h\psi_n,\Ab_n;\Omega)\geq-E_1|\partial\Omega|\kappa_n\nonumber\\
-c_\alpha
E_2|D|\, [\mu(\kappa_n)]^2\kappa_n+o(\mu^2_n\kappa_n)\,.
\end{equation}
Actually, we have proved the following lower bound,
\begin{equation}\label{eq-hc2-lb:gl'}
\liminf_{n\to\infty}\frac{\mathcal
E_0(h\psi_n,\Ab_n;\Omega)}{\mu^2(\kappa_n)\kappa_n}\geq\left\{
\begin{array}{ll} -E_2|D|c_{\alpha}&{\rm
if}~\mu_0=+\infty\\
-E_1|\partial\Omega|\mu_0^{-2}-E_2|D|c_{\alpha}&{\rm
if}~\mu_0\in(0,\infty)\,.
\end{array}\right.\end{equation}
Since the term on the left side in \eqref{eq-hc2-lb:gl'}
is independent from $\alpha$,
we get by taking $\alpha\to0_+$ on both sides (recall that
$c_\alpha=1+\alpha(1-\alpha)$),
\begin{equation}\label{eq-hc2-lb:gl}
\liminf_{n\to\infty}\frac{\mathcal
E_0(h\psi_n,\Ab_n;\Omega)}{\mu^2(\kappa_n)\kappa_n}\geq\left\{
\begin{array}{ll} -E_2|D|&{\rm
if}~\mu_0=\infty\\
-E_1|\partial\Omega|\mu_0^{-2}-E_2|D|&{\rm
if}~\mu_0\in(0,\infty)\,,
\end{array}\right.\end{equation}
which contradicts the upper bound in
\eqref{eq-hc2-assumption-n}. Therefore, the conclusion of
Theorem~\ref{eq-hc2-lb-loc} holds true.\hfill$\Box$

Recalling that $\mathcal E(\psi,\Ab;\Omega)=\mathcal
E_0(\psi,\Ab;\Omega)+(\kappa H)^2\int_\Omega|\curl\Ab-1|^2\,dx$, we
get as an immediate corollary of Theorem \ref{thm-hc2-FK-ub} and the
proof of Theorem~\ref{thm-hc2-lb-loc}:

\begin{corol}\label{corol-thm-hc2-lb}
Assume that $H=\kappa-\mu(\kappa)\sqrt{\kappa}$ with
$\displaystyle\lim_{\kappa\to\infty}\frac{\mu(\kappa)}{\sqrt{\kappa}}=0$.
Then, for any minimizer $(\psi,\Ab)$ of \eqref{eq-hc2-GL}, the
following asymptotic estimate holds:
$$\kappa^2H^2\int_{\Omega}|\curl\Ab-1|^2\,dx=o\left(\max([\mu(\kappa)]_+^2,1)\kappa\right)\quad{\rm
as}~\kappa\to\infty\,.$$
\end{corol}

\section{Proof of the energy estimates}\label{sec-hc2-proofs}

We proceed in this section to complete the proofs of
Theorem~\ref{thm-hc2-FK1} and Corollary~\ref{corol-hc2-FK1}.

We start by a localization estimate.

\begin{lem}\label{lem-hc2-IMS}
Assume $H=\kappa-\mu(\kappa)\sqrt{\kappa}$ such that
$$\lim_{\kappa\to\infty}\frac{\mu(\kappa)}{\sqrt{\kappa}}=0\,.$$
Then, for any minimizer $(\psi,\Ab)$ of \eqref{eq-hc2-GL} and any
open, smooth domain  $D\subset\Omega$, we have as $\kappa\to\infty$:
\begin{align}
  \label{eq:10}
\mathcal E(\psi,\Ab;D)-\mathcal
E(f\psi,\Ab;D)-\int_D|\nabla f|^2|\psi|^2\,dx
&+\frac{\kappa^2}2\int_D(1-f^2)^2|\psi|^4\,dx\nonumber \\
&=- {\rm Re}\int_{\partial \Omega} |\psi|^2 \overline{ f} \;\nu
\cdot \nabla f\,d\sigma+o(\kappa)\,.
\end{align}
Here $\nu$ is the unit inward normal vector of $\partial\Omega$ and
$f$ is any function such that,
 $$\nabla f\in L^\infty(\R^2)\,,\quad {\rm supp}\,f\subset \overline D\,.$$
\end{lem}
\begin{proof}[Proof of Lemma~\ref{lem-hc2-IMS}]~\\
Integrating by parts, we get the following localization
formula,
\begin{align}
  \label{eq:9}
  \int_\Omega|(\nabla-i\kappa H)f\psi|^2\,dx&={\rm Re}\int_\Omega
-(\nabla-i\kappa H\Ab)^2\psi\,\overline {f^2\psi}\,dx\nonumber \\
&\quad-{\rm Re}\int_{\partial \Omega} |\psi|^2 \overline{ f} \;\nu \cdot
\nabla f\,d\sigma +\int_\Omega|\nabla f|^2|\psi|^2\,dx\,.
\end{align}
Using the equation for
$\psi$ in \eqref{eq-hc2-GLeq} and the assumption ${\rm
supp}\,f\subset \overline D$, we get further,
\begin{multline}\label{eq-hc2-En1}
\mathcal
E(f\psi,\Ab;D)=\kappa^2\int_\Omega\left(\frac12f^2-1\right)f^2
|\psi|^4\,dx+\int_\Omega|\nabla f|^2|\psi|^2\,dx -{\rm
Re}\int_{\partial \Omega} |\psi|^2 \overline{ f} \;\nu \cdot \nabla
f\,d\sigma\\+(\kappa H)^2\int_{D}|\curl\Ab-1|^2\,dx\,.
\end{multline}
Similarly, we get,
\begin{multline}\label{eq-hc2-En2}
\mathcal E(\psi,\Ab;D)=
-\frac{\kappa^2}2\int_D|\psi|^4\,dx+(\kappa
H)^2\int_{D}|\curl\Ab-1|^2\,dx\\
+ \int_{\overline D\cap\partial\Omega}\overline\psi\,\nu\cdot(\nabla-i\kappa
H\Ab)\psi\,d\sigma +\int_{\Omega\cap\partial
D}\overline\psi\,\nu\cdot(\nabla-i\kappa H\Ab)\psi\,d\sigma\,.
\end{multline}
Combining \eqref{eq-hc2-En1} and \eqref{eq-hc2-En2}, we need only
establish that
$$\int_{\overline D\cap\partial\Omega}\overline\psi\,\nu\cdot(\nabla-i\kappa
H\Ab)\psi\,d\sigma +\int_{\Omega\cap\partial
D}\overline\psi\,\nu\cdot(\nabla-i\kappa
H\Ab)\psi\,d\sigma=o(\kappa)\,,
$$
as $\kappa\to\infty$.

Thanks to the boundary condition in \eqref{eq-hc2-GLeq}, the
integral over $\overline D\cap\partial\Omega$ vanishes. So we only consider the
integral over $\Omega\cap\partial D$. To that end we write,
\begin{multline*}
\int_{\Omega\cap\partial D}\overline\psi\,\nu\cdot(\nabla-i\kappa
H\Ab)\psi\,d\sigma =\int_{\{x\in\Omega\cap\partial D~:~{\rm
dist}(x,\partial\Omega)\leq  g_1(\kappa)\}}
\overline\psi\,\nu\cdot(\nabla-i\kappa H\Ab)\psi\,d\sigma\\+
\int_{\{x\in\Omega\cap\partial D~:~{\rm dist}(x,\partial\Omega)\geq
g_1(\kappa)\}} \overline\psi\,\nu\cdot(\nabla-i\kappa
H\Ab)\psi\,d\sigma \,.\end{multline*} Here $g_1(\kappa)$ is any
positive function such that $g_1(\kappa)\to0$ and $\kappa
g_1(\kappa)\to\infty$ as $\kappa\to\infty$. Invoking  the results
of Lemmas~\ref{lem-hc2-FoHe} and \ref{lem-FH-jems}, we deduce that
\begin{equation}\label{eq-hc2-D-bnd-cont}
\int_{\Omega\cap\partial D}\overline\psi\,\nu\cdot(\nabla-i\kappa
H\Ab)\psi\,d\sigma=o(\kappa)\quad{\rm
as}~\kappa\to\infty\,.\end{equation}
\end{proof}

\begin{proof}[Proof of Theorem~\ref{thm-hc2-FK1}]~\\
Thanks to Theorem~\ref{thm-hc2-Pan}, we may assume the condition 
$\displaystyle\liminf_{\kappa\to\infty}\mu(\kappa)>-\infty$.

Let us consider a partition of unity on $\R$,
$$h_1^2+h_2^2=1\quad{\rm in}~\R\,,\quad {\rm
  supp}\,h_1\subset(-1,\infty)\,,\quad{\rm
  supp}\,h_2\subset(-\infty,-\frac12)\,.
$$
Let $m=m(\kappa)\in(0,1)$ be a parameter that will be specified
later. Define the `signed' distance to the boundary between $D$ and
$\Omega \setminus \overline D$,
$$t_D(x)=\left\{\begin{array}{ll}
-{\rm dist}(x,\Gamma)&{\rm if}~x\in D\,,\\
{\rm dist}(x,\Gamma)&{\rm if}~x\not\in D\,.\end{array}\right., \qquad
\text{ with } \Gamma := \partial D \setminus\partial\Omega\,.$$
For any $x\in\overline\Omega$, we define,
$$\varphi_{1,m}(x)=h_1\left(\frac{t_D(x)}{m}\right)\,\psi(x)\,,\quad
\varphi_{2,m}(x)=h_2\left(\frac{t_D(x)}{m}\right)\,\psi(x)\,.$$
Then, it results from the IMS decomposition formula,
\begin{equation}\label{eq-hc2-D-IMS}
\mathcal E(\psi,\Ab;\Omega)\geq\mathcal
E(\varphi_{1,m},\Ab;\Omega)+ \mathcal
E(\varphi_{2,m},\Ab;\Omega)-\frac{C}{m^2}\int_\Omega|\psi|^2\,dx\,.
\end{equation}

We shall assume the following condition on $m=m(\kappa)$,
\begin{equation}\label{eq-hc2-D-eta}
m\ll1\quad{\rm and}\quad m^{-1}+ m^{-2}\zeta^2\ll
\max(1,[\mu(\kappa)]_{+}^2)\kappa\quad{\rm
  as}~\kappa\to\infty\,.\end{equation}
Here $\zeta(\kappa)=\max\{|1-\frac{\kappa}H|^{1/2},\kappa^{-1/4}\}$ as
previously. The choice $m=\frac1{\sqrt{\kappa}}$ fulfills the
condition in \eqref{eq-hc2-D-eta}.

Invoking Corollary~\ref{cor-hc2:L2} and the upper bound of
Theorem~\ref{thm-hc2-FK-ub}, we get under the condition
\eqref{eq-hc2-D-eta},
\begin{equation}\label{eq-hc2-D-IMS*}
\mathcal E(\varphi_{1,m},\Ab;\Omega)+ \mathcal
E(\varphi_{2,m},\Ab;\Omega)\leq -\mathcal A(\mu(\kappa);\Omega)
+o\left(\max(1,[\mu(\kappa)]_{+}^2)\kappa\right)\,,
\end{equation}
where, for a subdomain $V\subset\Omega$, we define,
\begin{equation}\label{eq-hc2-def-V}
\mathcal A(\mu(\kappa);V)=
\left(E_1|\overline
V\cap\partial\Omega|+[\mu(\kappa)]_+^2E_2|V|\right)\kappa\,.
\end{equation}
Notice that $\varphi_{1,m}$ has support in $\overline U_m$, where
$$U_m=\{x\in\Omega~:~{\rm dist}(x, U))<m\}\,,\quad U=\Omega\setminus
\overline D\,.$$ Applying Theorem~\ref{thm-hc2-lb-loc},
we get the following lower bound,
\begin{equation}\label{eq-hc2-lb-222}
\mathcal E(\varphi_{1,m},\Ab;\Omega)\geq -\mathcal
A(\mu(\kappa);\Omega\setminus \overline
D)+o\left(\max(1,[\mu(\kappa)]_{+}^2)\kappa\right)\quad{\rm
as}~\kappa\to\infty\,. \end{equation}
Substituting \eqref{eq-hc2-lb-222} in \eqref{eq-hc2-D-IMS*}, we also
get the following upper bound
\begin{equation}\label{eq-hc2-D-claim} \mathcal E(\varphi_{2,m},\Ab;\Omega)\leq
-\mathcal A(\mu(\kappa);
D)+o\left(\max(1,[\mu(\kappa)]_{+}^2)\kappa\right)\,.
\end{equation}

In order to finish the proof of Theorem~\ref{thm-hc2-FK1}, it is
sufficient to show for an arbitrary smooth domain $D\subset\Omega$,
\begin{equation}\label{eq-hc2-D-claim*}
\mathcal E(\psi,\Ab;D)\leq -\mathcal
A(\mu(\kappa),D)+o\left(\max(1,[\mu(\kappa)]_{+}^2)\kappa\right)\quad{\rm
as}~\kappa\to\infty\,,\end{equation} and
\begin{equation}\label{eq-hc2-D-claim**}
\mathcal E(\psi,\Ab;\Omega\setminus D)\geq -\mathcal
A(\mu(\kappa),\Omega\setminus
D)+o\left(\max(1,[\mu(\kappa)]_{+}^2)\kappa\right)\quad{\rm
as}~\kappa\to\infty\,.\end{equation} Let us prove
\eqref{eq-hc2-D-claim*}. Notice that $\varphi_{2,m}$ has support in
$\overline D$. Invoking Lemma~\ref{lem-hc2-IMS} together
with Corollary~\ref{cor-hc2:L2}, we get (thanks in particular to the
condition \eqref{eq-hc2-D-eta} on $m$),
$$\mathcal E(\psi,\Ab;D)\leq\mathcal E(\varphi_{2,m},\Ab;D)+
o\left(\max(1,[\mu(\kappa)]_{+}^2)\kappa\right)\,.$$ Using
\eqref{eq-hc2-D-claim}, we see that \eqref{eq-hc2-D-claim*} actually
holds.

Let us prove now \eqref{eq-hc2-D-claim**}. We have the natural
decomposition of the energy,
$$\mathcal E(\psi,\Ab;\Omega)=\mathcal E(\psi,\Ab;D)+\mathcal E(\psi,\Ab;\Omega\setminus
D)\,.$$ Using the lower bound in Theorem~\ref{thm-hc2-lb}, we deduce
that,
$$\mathcal E(\psi,\Ab;D)+\mathcal E(\psi,\Ab;\Omega\setminus
D)\geq-\mathcal A(\mu(\kappa),\Omega)+
o\left(\max(1,[\mu(\kappa)]_{+}^2)\kappa\right)\,.$$ Inserting the
established upper bound \eqref{eq-hc2-D-claim*} in the left side
above, we get the lower bound in \eqref{eq-hc2-D-claim**}.
\end{proof}

\begin{proof}[Proof of Corollary~\ref{corol-hc2-FK1}]~\\
Consider any open domain $D\subset\Omega$. Let us multiply the G-L
equation \eqref{eq-hc2-GLeq} for $\psi$ by $\overline\psi$ and
integrate over $\Omega$. Integrating by parts and using
Corollary~\ref{corol-thm-hc2-lb}, we obtain,
\begin{multline*}\mathcal E
(\psi,\Ab;D)=-\frac{\kappa^2}2\int_D|\psi|^4\,dx+
\int_{\overline\Omega\cap\partial
D}\overline\psi\,\nu_{D}\cdot(\nabla-i\kappa
H\Ab)\psi\,\,d\sigma\\+o\left(\max(1,[\mu(\kappa)]_{+}^2)\kappa\right)\,,\quad{\rm
as }~\kappa\to\infty\,.\end{multline*} Using
Lemmas~\ref{lem-hc2-FoHe} and ~\ref{lem-FH-jems}, we get (see the
proof of \eqref{eq-hc2-D-bnd-cont}),
$$\int_{\overline\Omega\cap\partial
D}\overline\psi\,\nu_D\cdot(\nabla-i\kappa
H\Ab)\psi\,d\sigma=o(\kappa)\quad{\rm as}~\kappa\to\infty\,.$$ In
particular we have,
$$\mathcal E
(\psi,\Ab;D)=-\frac{\kappa^2}2\int_D|\psi|^4\,dx+o\left(\max(1,[\mu(\kappa)]_{+}^2)\kappa\right)\,,\quad{\rm
as }~\kappa\to\infty\,.$$ Implementing the asymptotic expansion of
Theorem~\ref{thm-hc2-FK1}, we obtain,
$$\frac{\kappa^2}2\int_{D}|\psi|^4\,dx=\mathcal A(\mu(\kappa),D)+o\left(\max(1,[\mu(\kappa)]_{+}^2)\kappa\right)
\,,\quad{\rm as }~\kappa\to\infty\,.$$ Coming back to the definition
of $\mathcal A(\mu(\kappa),D)$  in \eqref{eq-hc2-def-V}, we get the
result of Corollary~\ref{corol-hc2-FK1}.
\end{proof}

\section*{Acknowledgements}
The authors were supported by the European Research Council under the European Community's
Seventh Framework Programme (FP7/2007-2013)/ERC grant agreement
n$^{\rm o}$ 202859. SF is also supported by the Danish
Research Council and the Lundbeck Foundation.

\end{document}